\documentclass{article}%
\usepackage{amsmath,extarrows}
\usepackage{amsfonts}
\usepackage{amssymb}
\usepackage{amsthm}
\usepackage{graphicx}
\usepackage{fullpage}
\usepackage{mathdots}
\usepackage{color}
\usepackage{hyperref}
\usepackage{setspace}%
\usepackage{amsmath}%
\providecommand{\U}[1]{\protect\rule{.1in}{.1in}}
\allowdisplaybreaks
\newtheorem{theorem}{Theorem}

\newtheorem{definition}[theorem]{Definition}
\newtheorem{example}[theorem]{Example}

\newtheorem{lemma}[theorem]{Lemma}
\newtheorem{notation}[theorem]{Notation}

\newtheorem{remark}[theorem]{Remark}

\begin{document}

\title{ Subresultant of several univariate polynomials}
\author{Hoon Hong\\Department of Mathematics, North Carolina State University\\Box 8205, Raleigh, NC 27695, USA\\hong@ncsu.edu\\[10pt] Jing Yang\thanks{Corresponding author.}\\SMS--HCIC--School of Mathematics and Physics, \\Guangxi Minzu University, Nanning 530006, China\\yangjing0930@gmail.com\\[-10pt]}
\date{}
\maketitle

\begin{abstract}
Subresultant of two univariate polynomials is a fundamental object in
computational algebra and geometry with many applications (for instance,
parametric GCD and parametric multiplicity of roots). In this paper, we
generalize the theory of subresultants of two polynomials to arbitrary number of
polynomials, resulting in \emph{multi}-polynomial subresultants. Specifically,

\begin{enumerate}
\item we propose a definition of multi-polynomial subresultants, which is an  expression  in terms of roots;

\item we illustrate the  usefulness of the proposed definition via the following two fundamental applications:

\begin{itemize}
\item parametric GCD of multi-polynomials, and

\item parametric multiplicity of roots of a   polynomial;

\end{itemize}

\item we provide several expressions for the multi-polynomials subresultants in
terms of coefficients, for computation.

\end{enumerate}
\end{abstract}

\section{Introduction}

Subresultant for two univariate polynomials is a fundamental object in
computational algebra and geometry with numerous applications in science and
engineering. Due to its importance, there have been extensive research on
underlying theories, extensions and applications (just list a few
\cite{1853_Sylvester,
2003_Lascoux_Pragacz,
1907_Bocher,
1968_Householder,
1967_Collins,
1971_BARNETT,
1971_Brown_Traub,
1997_Hong,
1998_Wang,
Hong:99a,
2000_Wang,
2004_WANG_HOU,
2008_Akira,
2004_DIAZ_TOCA_GONZALEZ_VEGA,
2007_DAndrea_Hong_Krick_Szanto,
2006_Andrea_Krick_Szanto,
2009_DAndrea_Hong_Krick_Szanto,
2015_DAndrea_Krick_Szanto,
2016_Krick_Szanto_Valdettaro,
2017_Bostan_DAndrea_Krick_Szanto_Valdettaro,
2019_DANDREA_Krick_Szanto_Valdettaro,
2020_Bostan_Krick_Szanto_Valdettaro},
\cite{
Li:98,
Hong:99b,
Hong:2000a,
2001_DAndrea_Dickenstein,
2013_DAndrea_Krick_Szanto,
2020_Roy_Szpirglas,
2021_Cox_DAndrea,
2013_Jaroschek} and
\cite{
2008_Szanto,
2017_Imbach_Moroz_Pouget,
2018_Perrucci_Roy,
2020_Roy_Szpirglas}).

 One often needs to consider more than two univariate polynomials, in various applications, for instance, in the following two fundamental applications:         parametric GCD of several  polynomials and  parametric multiplicity of roots of a   polynomial.   A usual way is to  apply subresultants  to a pair of input or intermediate polynomials repeatedly, resulting in ``nested'' subresultants, {\it i.e.},  subresultants of subresultants of ... and so on.   This process, however, often produces extraneous or  irrelevant factors, complicating theory  and in turn applications.

Hence one wonders whether    there is a way to generalize the theory of subresultant to more than two polynomials  such that nestings are avoided,  which would in turn   simplify subsequent theory developments and applications.
In this article, we provide such a generalized theory. Specifically,

\begin{enumerate}
\item {\em We propose a definition of multi-polynomial subresultants, which is an expression  in terms of roots}.

Naturally there are two different approaches to generalizing subresultants of two polynomials to arbitrary number of polynomials:
\begin{itemize}
\item[] (A1) generalizing expressions of subresultants in {\em coefficients}, such as minor of a Sylvester matrix.
\item[] (A2) generalizing expressions of subresultants in {\em roots}, such as Poisson formula.
\end{itemize}
After trying both approaches, we finally chose (A2) because the resulting definition would  encode   desired geometric information more explicitly and
such explicit geometric information would  greatly facilitates the developments of applications.

See Definition~\ref{def:roots} for the details.

\item {\em We illustrate the  usefulness of the proposed definition via the following two fundamental applications}.

\begin{itemize}
\item Parametric GCD of multi-polynomials.

Consider several parametric univariate polynomials  (such as polynomials with parametric coefficients).  The GCD of the polynomials will of course depend on the values of the parameters. Hence the parametric GCD problem asks for several conditions on parameters and corresponding expressions for GCD in terms of the parameters.
We show that  the conditions and the corresponding GCD can be simply/elegantly written  in terms of the proposed multi-polynomial subresultants.

See Theorem~\ref{thm:gcd} for details.

\item Parametric multiplicity of roots of a   polynomial.

Consider a parametric univariate polynomial (such as a polynomial with parametric coefficients). The multiplicity structure  (the vector of multiplicities of the distinct roots) of the  polynomial   will depend on the values of the parameters. Hence the parametric  multiplicity problem asks for a  condition on parameters for each multiplicity structure. We show that the conditions can be also simply/elegantly written in terms of the proposed multi-polynomial subresultants. .

See Theorem~\ref{thm:multiplicity} for details.
\end{itemize}

\item {\em We provide several expressions for multi-polynomials subresultants in
terms of coefficients}.

In the literature, there are roughly three types \ of expressions of subresultants in coefficients and their combinations: {\em Sylvester}-type, {\em B\'{e}zout-type}  and {\em Barnnet}-type.  We generalize all the three types to the multi-polynomial case.

See Theorem~\ref{thm:coefficients} for details.

\end{enumerate}


\medskip\noindent\textit{Comparison with related works}: There have been
numerous works on subresultants. We will focus on efforts in generalizing the
notion and theory of the subresultant of two univariate polynomials.

\begin{itemize}
\item In~\cite{
1853_Sylvester,
Hong:99a,
Hong:99b,
Hong:2000a,
2003_Lascoux_Pragacz,
2006_Andrea_Krick_Szanto,
2007_DAndrea_Hong_Krick_Szanto,
2009_DAndrea_Hong_Krick_Szanto,
2013_DAndrea_Krick_Szanto,
2015_DAndrea_Krick_Szanto,
2017_Bostan_DAndrea_Krick_Szanto_Valdettaro},
the subresultants are
expressed as rational function in roots.

We extended it to multi-polynomial
case. Difference is that the index for the subresultants are not just a
natural number, but a {\em tuple} of natural numbers.

\item In \cite{1978_Vardulakis_Stoyle}, the authors generalized the
following well-known property of the subresultant of two univariate
polynomials: the degree of the GCD of two polynomials is determined by the
rank of the corresponding Sylvester matrix . They extended Sylvester matrix to
several univariate polynomials so that a similar property holds.

In this paper, we generalize another well-known property of the subresultant
of two univariate polynomials: the degree of the GCD of two polynomials is
determined by vanishing of subresultants and the GCD is given by the
corresponding subresultant polynomial. We extend subresultant to several
polynomials so that a similar property holds.

\item In \cite{1971_BARNETT}, the author considered the problem of computing
the GCD for several univariate polynomials and gave an algorithm to write down
the expression of GCD by extending the Barnett matrix from bi-polynomial case
to multi-polynomial case.

In this paper, we further refine the result in the sense that we provide an
explicit formula to write down GCD in terms of coefficients of the given polynomials.

\item In \cite{2004_WANG_HOU}, the authors considered the B\'{e}zout-type
subresultant for two univariate polynomials and established the connection
between their pseudo-remainders and the subresultants; in
\cite{2004_DIAZ_TOCA_GONZALEZ_VEGA}, the authors presented more alternative
expressions which express subresultants in terms of some minors of matrices
different from the Sylvester matrix.

In this paper, we provide three different expressions for subresultant of
several univariate polynomials, including Sylvester-type, B\'{e}zout-type and Barnett-type.

\item \ {In \cite{1995_Chardin, 2001_DAndrea_Dickenstein, 2006_Andrea_Krick_Szanto, 1990_Gonzalez_Vega, 2008_Szanto, 2010_Szanto}, the authors generalized the subresultant of two
univariate polynomials to that of\emph{\ multivariate\/} polynomials while
constraining the number of polynomials to be at most one more than the number
of variables.
In \cite{Li:98,Hong:99b,Hong:2000a},
the authors generalized the subresultant of two polynomials to that of two Ore polynomials.
}

{In this paper, we take another route for generalization. We allow
\emph{arbitrary} number of polynomials while staying univariate. Of course, a  natural challenge for future work is to  generalize further to arbitrary number
of multivariate and/or Ore polynomials.
}

\end{itemize}

The paper is structured as follows. In Section \ref{sec:expression_in_roots},
we give a ``geometric" concept (definition) of subresultant for several
univariate polynomials (multi-polynomial subresultant) in terms of roots of
one polynomials among them. In Section \ref{sec:application_parametric_gcd},
we show that the parametric GCD of several univariate polynomials can be
expressed in terms of the multi-polynomial subresultant. In Section
\ref{sec:application_parametric_multiplicity}, we show that the parametric
multiplicity of one univariate polynomial can be expressed in terms of
subresultants of the given polynomial and its derivatives. For the sake of
computation, we study the expression of subresultant for multi-polynomials in
coefficients In Section \ref{sec:expression_in_coefficients} we identify
three different types of determinantal expressions in coefficients. In Section
\ref{sec:conclusion}, we summarize the main results and mention a natural extension.

\section{Definition of multi-polynomial subresultants in terms of roots}
\label{sec:expression_in_roots}



\begin{notation}
\label{notation1} $\ $

\begin{enumerate}
\item $F=\left(  F_{0},F_{1},\ldots,F_{t}\right)  \subset\mathbb{C}\left[
x\right]  $ $\,$where $t\geq1$.

\item $d_{i}=\deg F_{i}$

\item $\alpha_{1},\ldots,\alpha_{d_{0}}$ are the complex roots of $F_{0}$.

\item $\delta=\left(  \delta_{1},\ldots,\delta_{t}\right)  \mathbb{\in
N}_{\geq0}^{t}\ $such that $\left\vert \delta\right\vert \leq d_{0}$, where
$\left\vert \delta\right\vert =\delta_{1}+\cdots+\delta_{t}$.
\end{enumerate}
\end{notation}

\noindent With the above notations, we define the concept of subresultant
polynomial in terms of roots.

\begin{definition}
[Subresultant polynomial in terms of roots]\label{def:roots} We define the
$\delta$-th \emph{subresultant polynomial} $S_{\delta}$ of $F$\ by
\[
S_{\delta}\left(  F\right)  :=a_{0d_{0}}^{\delta_{0}}\cdot\frac{\det\left[
\begin{array}
[c]{rrr|c}%
\alpha_{1}^0 F_{1}\left(  \alpha_{1}\right)   & \cdots &
\alpha_{d_0}^0 F_{1} \left(  \alpha_{d_{0}}\right)   & \\
\vdots~~~~~~~~ &  & \vdots~~~~~~~~ & \\
\alpha_{1}^{\delta_{1}-1} F_{1}(\alpha_{1}) & \cdots &
\alpha_{d_0}^{\delta_{1}-1} F_{1}(\alpha_{d_{0}}) & \\\hline
\vdots~~~~~~~~ &  & \vdots~~~~~~~~ & \\
\vdots~~~~~~~~ &  & \vdots~~~~~~~~ & \\\hline
\alpha_{1}^0 F_{t}\left(  \alpha_{1}\right)   & \cdots &
\alpha_{d_0}^0 F_{t} \left(  \alpha_{d_{0}}\right)   & \\
\vdots~~~~~~~~ &  & \vdots~~~~~~~~ & \\
\alpha_{1}^{\delta_{t}-1} F_{t}(\alpha_{1}) & \cdots &
\alpha_{d_0}^{\delta_{t}-1} F_{t}(\alpha_{d_{0}}) & \\\hline
\alpha_{1}^0~~~~~~~   & \cdots & \alpha_{d_{0}}^0~~~~~~~  & x^0\\
\vdots~~~~~~~~ &  & \vdots~~~~~~~~ & \vdots\\
\alpha_{1}^{\varepsilon-1}~~~~~~   & \cdots &\alpha_{d_{0}}^{\varepsilon-1}~~~~~~   & x^{\varepsilon-1}%
\end{array}
\right]  }{\det\left[
\begin{array}
[c]{rrr}%
\alpha_{1}^0   & \cdots & \alpha_{d_{0}}^0\\
\vdots &   & \vdots\\
\alpha_{1}^{d_{0}-1}   & \cdots &\alpha_{d_{0}}^{d_0-1}
\end{array}
\right]  }%
\]
where

\begin{enumerate}
\item $a_{0d_{0}}$ is the leading coefficient of $F_{0}$ which is assumed to
be nonzero;

\item $\varepsilon$ is chosen so that the matrix becomes square, that is,
\[
1+d_{0}=\left\vert \delta\right\vert +\varepsilon
\]
in other words, $\varepsilon=1+d_{0}-\left\vert \delta\right\vert$;

\item $\delta_{0}$ is equal to the degree of $\alpha_{1}$ in the fraction,
that is,
\begin{align*}
\delta_{0}  &  =\max(\delta_{1}+d_{1}-1,\ldots,\delta_{t}+d_{t}-1,\varepsilon
-1)-(d_{0}-1)\\
&  =\max(\delta_{1}+d_{1}-d_{0},\ldots,\delta_{t}+d_{t}-d_{0},\varepsilon
-d_{0})\\
&  =\max(\delta_{1}+d_{1}-d_{0},\ldots,\delta_{t}+d_{t}-d_{0},1-\left\vert
\delta\right|  )
\end{align*}
The purpose of introducing $a_{0d_{0}}^{\delta_{0}}$ is to make the expression
in coefficients be a polynomial in $a_{ij}$ where $0\le i\le t$ and $0\le j\le
d_{i}$.
\end{enumerate} We define the
$\delta$-th \emph{subresultant} $s_{\delta}\left(  F\right)  $ of $F$\ to be
the \textquotedblleft principal\textquotedblright\ leading coefficient of
$S_{\delta},$ that is, the coefficient for $x^{\varepsilon-1}$.
\end{definition}

\begin{remark}
It is very important to note that $S_{\delta}\left(  F\right)  $ is a
polynomial function in $\alpha_{1},\ldots,\alpha_{d_{0}}$, even though written
as a rational function, since the numerator is exactly divisible by the
denominator. Hence the above definition should be read as follows:

\begin{enumerate}
\item Treating $\alpha_{1},\ldots,\alpha_{d_{0}}$ as distinct indeterminates,
carry out the exact division obtaining a polynomial.

\item Treating $\alpha_{1},\ldots,\alpha_{d_{0}}$ as numbers, evaluate the
resulting polynomial.
\end{enumerate}
\end{remark}

\begin{remark}
Let $\operatorname*{sres}\nolimits_{i}\left(  F_{0},F_{1}\right) $ stand for
the \textquotedblleft classical\textquotedblright\ $i$-th subresultant polynomial of
$F_{0}$ and $F_{1}$. Under the new notion, we have
\begin{align*}
\operatorname*{sres}\nolimits_{0}\left(  F_{0},F_{1}\right)   &  =S_{\left(
d_{0}\right)  }\left(  F_{0},F_{1}\right) \\
\operatorname*{sres}\nolimits_{1}\left(  F_{0},F_{1}\right)   &  =S_{\left(
d_{0}-1\right)  }\left(  F_{0},F_{1}\right) \\
&  \vdots\\
\operatorname*{sres}\nolimits_{i}\left(  F_{0},F_{1}\right)   &  =S_{\left(
d_{0}-i\right)  }\left(  F_{0},F_{1}\right) \\
&  \vdots\\
\operatorname*{sres}\nolimits_{d_{0}}\left(  F_{0},F_{1}\right)   &
=S_{\left( 0\right)  }\left(  F_{0},F_{1}\right)
\end{align*}
The index scheme is changed. The reason is the "new" one is better for extension.
\end{remark}

\begin{example}
Let $F=(F_{0},F_{1},F_{2})$. Let
\begin{align*}
F_{0}  &  =a_{03}\left(  x-\alpha_{1}\right)  \left(  x-\alpha_{2}\right)
\left(  x-\alpha_{3}\right) \\
F_{1}  &  =a_{13}x^{3}+a_{12}x^{2}+a_{11}x+a_{10}\\
F_{2}  &  =a_{21}x+a_{20}%
\end{align*}
where $a_{03}\neq0$. We would like to determine $S_{\left(
1,1\right)  }\left(  F\right)  $. Note that
\begin{align*}
\varepsilon&=1+3-\left(
1+1\right)  =2\\
\delta_0&=\max(3+1-3,1+1-3,1-(1+1))=1
\end{align*}
 Thus%
\begin{align*}
&  S_{(1,1)}\left(  F\right) \\
&  =a_{03}^{}\cdot\frac{\det\left[
\begin{array}
[c]{lll|l}%
\alpha_{1}^{0}F_{1}(\alpha_{1}) & \alpha_{2}^{0}F_{1}(\alpha_{2}) & \alpha
_{3}^{0}F_{1}(\alpha_{3}) & \\\hline
\alpha_{1}^{0}F_{2}(\alpha_{1}) & \alpha_{2}^{0}F_{2}(\alpha_{2}) & \alpha
_{3}^{0}F_{2}(\alpha_{3}) & \\\hline
\alpha_{1}^{0} & \alpha_{2}^{0} & \alpha_{3}^{0} & x^{0}\\
\alpha_{1}^{1} & \alpha_{2}^{1} & \alpha_{3}^{1} & x^{1}%
\end{array}
\right]  }{\det\left[
\begin{array}
[c]{ccc}%
\alpha_{1}^{0} & \alpha_{2}^{0} & \alpha_{3}^{0}\\
\alpha_{1}^{1} & \alpha_{2}^{1} & \alpha_{3}^{1}\\
\alpha_{1}^{2} & \alpha_{2}^{2} & \alpha_{3}^{2}%
\end{array}
\right]  }\\
&  =a_{03}^{}\cdot\frac{\det\left[
{\setlength{\arraycolsep}{1.2pt}
\begin{array}
[c]{lll|l}%
\alpha_{1}^{0}\left(  a_{13}\alpha_{1}^{3}+a_{12}\alpha_{1}^{2}+a_{11}%
\alpha_{1}+a_{10}\right)  & \alpha_{2}^{0}\left(  a_{13}\alpha_{2}^{3}%
+a_{12}\alpha_{2}^{2}+a_{11}\alpha_{2}+a_{10}\right)  & \alpha_{3}^{0}\left(
a_{13}\alpha_{3}^{3}+a_{12}\alpha_{3}^{2}+a_{11}\alpha_{3}+a_{10}\right)  &
0\\\hline
\alpha_{1}^{0}\left(  a_{21}\alpha_{1}+a_{20}\right)  & \alpha_{2}^{0}\left(
a_{21}\alpha_{2}+a_{20}\right)  & \alpha_{3}^{0}\left(  a_{21}\alpha
_{3}+a_{20}\right)  & 0\\\hline
\alpha_{1}^{0} & \alpha_{2}^{0} & \alpha_{3}^{0} & x^{0}\\
\alpha_{1}^{1} & \alpha_{2}^{1} & \alpha_{3}^{1} & x^{1}%
\end{array}}
\right]  }{\det\left[
\begin{array}
[c]{ccc}%
\alpha_{1}^{0} & \alpha_{2}^{0} & \alpha_{3}^{0}\\
\alpha_{1}^{1} & \alpha_{2}^{1} & \alpha_{3}^{1}\\
\alpha_{1}^{2} & \alpha_{2}^{2} & \alpha_{3}^{2}%
\end{array}
\right]  }\\
&  =a_{03}^{}\cdot\frac{-\left(  \alpha_{2}-\alpha_{1}\right)  \left(  \alpha_{3}-\alpha
_{1}\right)  \left(  \alpha_{3}-\alpha_{2}\right)  \left(  a_{12}+a_{13}%
\alpha_{1}+a_{13}\alpha_{2}+a_{13}\alpha_{3}\right)  \left(  a_{21}%
x+a_{20}\right)  }{\left(  \alpha_{2}-\alpha_{1}\right)  \left(  \alpha
_{3}-\alpha_{1}\right)  \left(  \alpha_{3}-\alpha_{2}\right)  }\\
&  =-a_{03}^{}\left(  a_{12}+a_{13}\alpha_{1}+a_{13}\alpha_{2}+a_{13}\alpha_{3}\right)
\left(  a_{21}x+a_{20}\right)
\end{align*}

\end{example}

\section{Application: Parametric GCD}

\label{sec:application_parametric_gcd} In this section, we show an application
of multi-polynomial subresultant in representing the parametric GCD of
several univariate polynomials.

\subsection{Main results}

\begin{definition}
[Incremental Cofactor Degree]\label{notation:gcd}Let $F=\left(  F_{0}%
,F_{1},\ldots,F_{t}\right)  \subset\mathbb{C}\left[  x\right]  $ $\,$where
$t\geq1$. The \emph{incremental cofactor degree}, $\operatorname*{icdeg}%
\left(  F\right)  $, is defined by the following%
\[
\operatorname*{icdeg}\left(  F\right)  =\left(  \deg C_{1},\ldots,\deg
C_{t}\right)
\]
where $C_{i}$ is the incremental cofactor, that is,%
\[
C_{i}=\frac{\gcd\left(  F_{0},\ldots,F_{i-1}\right)  \ \ \ \ }{\gcd\left(
F_{0},\ldots,F_{i-1},F_{i}\right)  }%
\]

\end{definition}

\begin{remark}
\

\begin{enumerate}
\item In order to make the $\gcd$ unique, we take the usual convention that
$\gcd$ is monic.

\item The $\gcd$ of a single polynomial is defined to be the monic version of
the polynomial.
\end{enumerate}
\end{remark}

\begin{theorem}
[Parametric GCD in terms of subresultant polynomials]\label{thm:gcd}Let%
\[
\delta=\max\limits_{\substack{s_{\gamma}\left(  F\right)  \neq0}}\gamma
\]
where $\max$ is with respect to the ordering $\succ_{\text{glex}}$.\ Then we have

\begin{enumerate}
\item $\operatorname*{icdeg}\left(  F\right)  =\delta$

\item $\gcd\left(  F\right)  =\frac{S_{\delta}\left(  F\right)  }{s_{\delta
}\left(  F\right)  }$
\end{enumerate}
\end{theorem}

\begin{remark}
The ordering $\succ_{\text{glex}}$ is defined for two sequences, e.g.,
$\gamma$ and $\delta$, in $\mathbb{Z}_{\ge0}^{t}$. We say $\delta
\succ_{\text{glex}}\gamma$ if and only if one of the followings occurs:

\begin{enumerate}
\item $|\delta|>|\gamma|$;

\item $|\delta|=|\gamma|$ and there exists $i\le t$ such that $\delta
_{i}>\gamma_{i}$ and $\delta_{j}=\gamma_{j}$ for $j<i$.
\end{enumerate}
\end{remark}
\begin{remark}
In the above, the $\max$ always exists since $s_{\left(  0,\ldots,0\right)
}(F)=a_{0d_0}^{\max(d_1-d_0,\ldots,d_t-d_0,1)}\neq0$ from Definition~\ref{def:roots}.
\end{remark}

\begin{example}
Let $F=\left(  F_{0},F_{1},F_{2}\right)  $ where $\deg F_{0}=2$. We have the
following parametric gcd. For the sake of compactness, we present it using
\textquotedblleft nested if\textquotedblright.
\[%
\begin{array}
[c]{llcl}%
\text{if} & \ s_{\left(  2,0\right)  }\left(  F\right)  \neq0 & \text{then} &
\gcd(F)=\frac{S_{\left(  2,0\right)  }\left(  F\right)  }{{s}_{(2,0)}\left(
F\right)  }\\
\text{else if} & \ s_{\left(  1,1\right)  }\left(  F\right)  \neq0 &
\text{then} & \gcd(F)=\frac{S_{\left(  1,1\right)  }\left(  F\right)  }%
{{s}_{(1,1)}\left(  F\right)  }\\
\text{else if} & \ s_{\left(  0,2\right)  }\left(  F\right)  \neq0 &
\text{then} & \gcd(F)=\frac{S_{\left(  0,2\right)  }\left(  F\right)  }%
{{s}_{(0,2)}\left(  F\right)  }\\
\text{else if} & \ s_{\left(  1,0\right)  }\left(  F\right)  \neq0 &
\text{then} & \gcd(F)=\frac{S_{\left(  1,0\right)  }\left(  F\right)  }%
{{s}_{(1,0)}\left(  F\right)  }\\
\text{else if} & \ s_{\left(  0,1\right)  }\left(  F\right)  \neq0 &
\text{then} & \gcd(F)=\frac{S_{\left(  0,1\right)  }\left(  F\right)  }%
{{s}_{(0,1)}\left(  F\right)  }\\
\text{else} & \  & \text{then} & \gcd(F)=\frac{S_{\left(  0,0\right)  }\left(
F\right)  }{s_{\left(  0.0\right)  \left(  F\right)  }}%
\end{array}
\
\]

\end{example}

\begin{remark}
Previous works for generating the condition for parametric gcd computation
are mostly based on repeated gcd computation using subresultants for two
polynomials \cite{1993_Abramov,2012_CHEN_Moreno_Maza}.%
We compare them with the conditions given in this paper (Theorem~\ref{thm:gcd}).
\begin{itemize}
\item In the previous works, the polynomials in the conditions are principal coefficients of nested subresultants,
that is, subresultant of subresultant polynomials of .....  and so on.
In the current work, the polynomials are just principal multi-polynomial subresultants of the input.

\item The number of polynomials in the conditions given by the previous works and that given by the current work are compatible and both of them are one less than the number of possible $\delta$'s, i.e., the number of ordered $t$-weak-partitions\footnote[1]{By $t$-weak-partition, we mean a partition with $t$ parts where the  part $0$ is allowed.} of all natural numbers less than $d_0$. However, the polynomials in the conditions given by the previous work often have more extraneous factors.

\end{itemize}
\end{remark}


\subsection{Proof of Theorem \ref{thm:gcd} (Parametric GCD)}

We will use the following short hand notation.

\begin{notation}
[Vandermonde]$\ V(x_{1},\ldots,x_{n})=\det\left[
\begin{array}
[c]{ccc}%
x_{1}^0  & \cdots & x_{n}^0 \\
\vdots &  & \vdots\\
x_{1}^{n-1}  & \cdots & x_{n}^{n-1}
\end{array}
\right]  $
\end{notation}

\begin{notation}
[Incremental gcd]\

\begin{enumerate}
\item $G_{i}=\gcd\left(  F_{0},\ldots,F_{i}\right)  $.\ \ \ Note that
$G_{0}\ $is the monic version of $F_{0}$.

\item $e_{i}=\deg G_{i}$
\end{enumerate}
\end{notation}

\begin{lemma}
\label{lem:part_a}Let $\theta=\operatorname*{icdeg}\left(  F\right)  $.
Then we have
\[
s_{\theta}\left(  F\right)  \neq0\ \ \ \text{and\ \ \ }\gcd(F)=\frac
{S_{\theta}\left(  F\right)  }{s_{\theta}\left(  F\right)  }%
\]

\end{lemma}

\begin{proof}
In order to convey the main underlying ideas effectively, we will show the
proof for a particular case first. After that, we will generalize the ideas to
arbitrary cases.

\bigskip

\noindent\textbf{Particular case:} Consider the case $d=\left(  7,6,6\right)
\ $and $e=\left(  7,4,2\right)  $.

\begin{enumerate}
\item We treat $\alpha_{1},\ldots,\alpha_{7}$ as distinct indeterminates.

\item Note $\theta=\left(  e_{0}-e_{1},e_{1}-e_{2}\right)  =\left(
7-4,4-2\right)  =\left(  3,2\right)  $. Without loss of generality, we index
the roots as follows.
\[%
\begin{array}
[c]{ccr}%
G_{0} & = & \left(  x-\alpha_{1}\right)  \left(  x-\alpha_{2}\right)  \left(
x-\alpha_{3}\right)  \left(  x-\alpha_{4}\right)  \left(  x-\alpha_{5}\right)
\left(  x-\alpha_{6}\right)  \left(  x-\alpha_{7}\right) \\
G_{1} & = & \left(  x-\alpha_{4}\right)  \left(  x-\alpha_{5}\right)  \left(
x-\alpha_{6}\right)  \left(  x-\alpha_{7}\right) \\
G_{2} & = & \left(  x-\alpha_{6}\right)  \left(  x-\alpha_{7}\right)
\end{array}
\]
Then
\begin{align}
&  F_{1}(\alpha_{j})=0\ \ \text{for }j=4,5,6,7\label{eq:Fj_alphai_ex}\\
&  F_{2}(\alpha_{j})=0\ \ \text{for }j=6,7\nonumber
\end{align}

\item From the definition of subresultant (Definition \ref{def:roots}), we
have
\[
{S}_{\theta}(F)=a_{07}^{2}\cdot\frac{\det\left[
\begin{array}
[c]{ccc|cc|cc|c}%
\alpha_{1}^{0}F_{1}\left(  \alpha_{1}\right)  & \alpha_{2}^{0}F_{1}\left(
\alpha_{2}\right)  & \alpha_{3}^{0}F_{1}\left(  \alpha_{3}\right)  &
\alpha_{4}^{0}F_{1}\left(  \alpha_{4}\right)  & \alpha_{5}^{0}F_{1}\left(
\alpha_{5}\right)  & \alpha_{6}^{0}F_{1}\left(  \alpha_{6}\right)  &
\alpha_{3}^{0}F_{1}\left(  \alpha_{7}\right)  & \\
\alpha_{1}^{1}F_{1}\left(  \alpha_{1}\right)  & \alpha_{2}^{1}F_{1}\left(
\alpha_{2}\right)  & \alpha_{3}^{1}F_{1}\left(  \alpha_{3}\right)  &
\alpha_{4}^{1}F_{1}\left(  \alpha_{4}\right)  & \alpha_{5}^{1}F_{1}\left(
\alpha_{5}\right)  & \alpha_{6}^{1}F_{1}\left(  \alpha_{6}\right)  &
\alpha_{3}^{1}F_{1}\left(  \alpha_{7}\right)  & \\
\alpha_{1}^{2}F_{1}\left(  \alpha_{1}\right)  & \alpha_{2}^{2}F_{1}\left(
\alpha_{2}\right)  & \alpha_{3}^{2}F_{1}\left(  \alpha_{3}\right)  &
\alpha_{4}^{2}F_{1}\left(  \alpha_{4}\right)  & \alpha_{5}^{2}F_{1}\left(
\alpha_{5}\right)  & \alpha_{6}^{2}F_{1}\left(  \alpha_{6}\right)  &
\alpha_{3}^{2}F_{1}\left(  \alpha_{7}\right)  & \\\hline
\alpha_{1}^{0}F_{2}\left(  \alpha_{1}\right)  & \alpha_{2}^{0}F_{2}\left(
\alpha_{2}\right)  & \alpha_{3}^{0}F_{2}\left(  \alpha_{3}\right)  &
\alpha_{4}^{0}F_{2}\left(  \alpha_{4}\right)  & \alpha_{5}^{0}F_{2}\left(
\alpha_{5}\right)  & \alpha_{6}^{0}F_{2}\left(  \alpha_{6}\right)  &
\alpha_{3}^{0}F_{2}\left(  \alpha_{7}\right)  & \\
\alpha_{1}^{1}F_{2}\left(  \alpha_{1}\right)  & \alpha_{2}^{1}F_{2}\left(
\alpha_{2}\right)  & \alpha_{3}^{1}F_{2}\left(  \alpha_{3}\right)  &
\alpha_{4}^{1}F_{2}\left(  \alpha_{4}\right)  & \alpha_{5}^{1}F_{2}\left(
\alpha_{5}\right)  & \alpha_{6}^{1}F_{2}\left(  \alpha_{6}\right)  &
\alpha_{3}^{1}F_{2}\left(  \alpha_{7}\right)  & \\\hline
\alpha_{1}^{0} & \alpha_{2}^{0} & \alpha_{3}^{0} & \alpha_{4}^{0} & \alpha
_{5}^{0} & \alpha_{6}^{0} & \alpha_{7}^{0} & x^{0}\\
\alpha_{1}^{1} & \alpha_{2}^{1} & \alpha_{3}^{1} & \alpha_{4}^{1} & \alpha
_{5}^{1} & \alpha_{6}^{1} & \alpha_{7}^{1} & x^{1}\\
\alpha_{1}^{2} & \alpha_{2}^{2} & \alpha_{3}^{2} & \alpha_{4}^{2} & \alpha
_{5}^{2} & \alpha_{6}^{2} & \alpha_{7}^{2} & x^{2}%
\end{array}
\right]  }{V\left(  \alpha_{1},\ldots,\alpha_{7}\right)  }%
\]
where $a_{07}$ is the leading coefficient of $F_0$, and
\begin{align*}
\varepsilon&=1+d_{0}-\left\vert \theta\right\vert =1+7-(3+2)=3\\
\theta_0&=\max(d_1+\theta_1-d_0,d_1+\theta_1-d_0,1-|\theta|)\\
&=\max(6+3-7,6+2-7,1-(3+2))=2
\end{align*}

\item Substitution of (\ref{eq:Fj_alphai_ex}) into the above expression yields
the following determinant with a block lower-triangular structure:%
\begin{align*}
{S}_{\theta}(F)  &  =a_{07}^{2}\cdot\frac{\det\left[
\begin{array}
[c]{ccc|cc|ccc}%
\alpha_{1}^{0}F_{1}\left(  \alpha_{1}\right)  & \alpha_{2}^{0}F_{1}\left(
\alpha_{2}\right)  & \alpha_{3}^{0}F_{1}\left(  \alpha_{3}\right)  &  &  &  &
& \\
\alpha_{1}^{1}F_{1}\left(  \alpha_{1}\right)  & \alpha_{2}^{1}F_{1}\left(
\alpha_{2}\right)  & \alpha_{3}^{1}F_{1}\left(  \alpha_{3}\right)  &  &  &  &
& \\
\alpha_{1}^{2}F_{1}\left(  \alpha_{1}\right)  & \alpha_{2}^{2}F_{1}\left(
\alpha_{2}\right)  & \alpha_{3}^{2}F_{1}\left(  \alpha_{3}\right)  &  &  &  &
& \\\hline
\alpha_{1}^{0}F_{2}\left(  \alpha_{1}\right)  & \alpha_{2}^{0}F_{2}\left(
\alpha_{2}\right)  & \alpha_{3}^{0}F_{2}\left(  \alpha_{3}\right)  &
\alpha_{4}^{0}F_{2}\left(  \alpha_{4}\right)  & \alpha_{5}^{0}F_{2}\left(
\alpha_{5}\right)  &  &  & \\
\alpha_{1}^{1}F_{2}\left(  \alpha_{1}\right)  & \alpha_{2}^{1}F_{2}\left(
\alpha_{2}\right)  & \alpha_{3}^{1}F_{2}\left(  \alpha_{3}\right)  &
\alpha_{4}^{1}F_{2}\left(  \alpha_{4}\right)  & \alpha_{5}^{1}F_{2}\left(
\alpha_{5}\right)  &  &  & \\\hline
\alpha_{1}^{0} & \alpha_{2}^{0} & \alpha_{3}^{0} & \alpha_{4}^{0} & \alpha
_{5}^{0} & \alpha_{6}^{0} & \alpha_{7}^{0} & x^{0}\\
\alpha_{1}^{1} & \alpha_{2}^{1} & \alpha_{3}^{1} & \alpha_{4}^{1} & \alpha
_{5}^{1} & \alpha_{6}^{1} & \alpha_{7}^{1} & x^{1}\\
\alpha_{1}^{2} & \alpha_{2}^{2} & \alpha_{3}^{2} & \alpha_{4}^{2} & \alpha
_{5}^{2} & \alpha_{6}^{2} & \alpha_{7}^{2} & x^{2}%
\end{array}
\right]  }{V\left(  \alpha_{1},\ldots,\alpha_{7}\right)  }\\
&  =a_{07}^{2}\cdot\frac{\det\left[
\begin{array}
[c]{ccc}%
M_{1} &  & \\
\cdot & M_{2} & \\
\cdot & \cdot & N
\end{array}
\right]  }{V\left(  \alpha_{1},\ldots,\alpha_{7}\right)  }%
\end{align*}

where%
\[%
\begin{array}
[c]{lll}%
M_{1} & = & \left[
\begin{array}
[c]{ccc}%
\alpha_{1}^{0}F_{1}\left(  \alpha_{1}\right)  & \alpha_{2}^{0}F_{1}\left(
\alpha_{2}\right)  & \alpha_{3}^{0}F_{1}\left(  \alpha_{3}\right) \\
\alpha_{1}^{1}F_{1}\left(  \alpha_{1}\right)  & \alpha_{2}^{1}F_{1}\left(
\alpha_{2}\right)  & \alpha_{3}^{1}F_{1}\left(  \alpha_{3}\right) \\
\alpha_{1}^{2}F_{1}\left(  \alpha_{1}\right)  & \alpha_{2}^{2}F_{1}\left(
\alpha_{2}\right)  & \alpha_{3}^{2}F_{1}\left(  \alpha_{3}\right)
\end{array}
\right] \\
\  &  & \\
M_{2} & = & \left[
\begin{array}
[c]{cc}%
\alpha_{4}^{0}F_{2}\left(  \alpha_{4}\right)  & \alpha_{5}^{0}F_{2}\left(
\alpha_{5}\right) \\
\alpha_{4}^{1}F_{2}\left(  \alpha_{4}\right)  & \alpha_{5}^{1}F_{2}\left(
\alpha_{5}\right)
\end{array}
\right] \\
\  &  & \\
N & = & \left[
\begin{array}
[c]{cc|c}%
\alpha_{6}^{0} & \alpha_{7}^{0} & x^{0}\\
\alpha_{6}^{1} & \alpha_{7}^{1} & x^{1}\\
\alpha_{6}^{2} & \alpha_{7}^{2} & x^{2}%
\end{array}
\right]
\end{array}
\]

\item Note that $M_{1},M_{2},N$ are all square matrices.

\item By applying elementary properties of determinants to the above
expressions, we have%
\[
{S}_{\theta}(F)=a_{07}^{2}\cdot\frac{\det M_{1}\ \det M_{2}\ \det N}{V\left(  \alpha
_{1},\ldots,\alpha_{7}\right)  }%
\]
where%
\[%
\begin{array}
[c]{llll}%
\det M_{1} & = & V\left(  \alpha_{1},\alpha_{2},\alpha_{3}\right)
\ \prod\limits_{j=1}^{3}F_{1}\left(  \alpha_{j}\right)  & =V\left(  \alpha
_{1},\alpha_{2},\alpha_{3}\right)  F_{1}\left(  \alpha_{1}\right)
F_{1}\left(  \alpha_{2}\right)  F_{1}\left(  \alpha_{3}\right) \\
\det M_{2} & = & V\left(  \alpha_{4},\alpha_{5}\right)  \ \prod\limits_{j=4}%
^{5}F_{2}\left(  \alpha_{j}\right)  & =V\left(  \alpha_{4},\alpha_{5}\right)
F_{2}\left(  \alpha_{4}\right)  F_{2}\left(  \alpha_{5}\right) \\
\det N & = & V\left(  \alpha_{6},\alpha_{7}\right)  \ \gcd\left(  F\right)  &
\text{since }\gcd\left(  F\right)  =G_{2}=\left(  x-\alpha_{6}\right)  \left(
x-\alpha_{7}\right)
\end{array}
\]

\item Thus we have%
\[
S_{\theta}\left(  F\right)  =s_{\theta}\left(  F\right)  \ \gcd\left(
F\right)
\]
where%
\[
s_{\theta}\left(  F\right)  =a_{07}^{2}\cdot\frac{V\left(  \alpha_{1},\alpha_{2},\alpha
_{3}\right)  V\left(  \alpha_{4},\alpha_{5}\right)  V\left(  \alpha_{6}%
,\alpha_{7}\right)  }{V\left(  \alpha_{1},\ldots,\alpha_{7}\right)  }%
\ F_{1}\left(  \alpha_{1}\right)  F_{1}\left(  \alpha_{2}\right)  F_{1}\left(
\alpha_{3}\right)  \ \ F_{2}\left(  \alpha_{4}\right)  F_{2}\left(  \alpha
_{5}\right)  \ \
\]

\item By applying elementary properties of Vandermonde determinants, we have
\begin{align*}
&  \frac{V\left(  \alpha_{1},\alpha_{2},\alpha_{3}\right)  V\left(  \alpha
_{4},\alpha_{5}\right)  V\left(  \alpha_{6},\alpha_{7}\right)  }{V\left(
\alpha_{1},\ldots,\alpha_{7}\right)  }\\
&  =\frac{1}{\prod\limits_{j=1}^{3}\prod\limits_{k=4}^{7}\left(  \alpha
_{k}-\alpha_{j}\right)  \ \prod\limits_{j=4}^{5}\prod\limits_{k=6}^{7}\left(
\alpha_{k}-\alpha_{j}\right)  }\\
&  =\pm\frac{1}{\prod\limits_{j=1}^{3}\prod\limits_{k=4}^{7}\left(  \alpha
_{j}-\alpha_{k}\right)  \ \prod\limits_{j=4}^{5}\prod\limits_{k=6}^{7}\left(
\alpha_{j}-\alpha_{k}\right)  }\\
&  =\pm\frac{1}{G_{1}\left(  \alpha_{1}\right)  G_{1}\left(  \alpha
_{2}\right)  G_{1}\left(  \alpha_{3}\right)  \ G_{2}\left(  \alpha_{4}\right)
G_{2}\left(  \alpha_{5}\right)  }%
\end{align*}

\item Thus
\begin{align*}
s_{\theta}\left(  F\right)  &=\pm a_{07}^{2}\cdot\frac{F_{1}\left(  \alpha_{1}\right)  }%
{G_{1}\left(  \alpha_{1}\right)  }\frac{F_{1}\left(  \alpha_{2}\right)
}{G_{1}\left(  \alpha_{2}\right)  }\frac{F_{1}\left(  \alpha_{3}\right)
}{G_{1}\left(  \alpha_{3}\right)  }\frac{F_{2}\left(  \alpha_{4}\right)
}{G_{2}\left(  \alpha_{4}\right)  }\frac{F_{2}\left(  \alpha_{5}\right)
}{G_{2}\left(  \alpha_{5}\right)  }\\
&=\pm a_{07}^{2}\cdot H_{1}\left(  \alpha_{1}\right)
\ H_{1}\left(  \alpha_{2}\right)  \ H_{1}\left(  \alpha_{3}\right)
\ H_{2}\left(  \alpha_{4}\right)  \ H_{2}\left(  \alpha_{5}\right)  \
\end{align*}
where $H_{i}=F_{i}/G_{i}$.

\item From now on, we treat $\alpha_{1},\ldots,\alpha_{7}$ as numbers.

\item Note that $s_{\theta}\left(  F\right)  \ \neq0$.

\item Thus we also have $\gcd\left(  F\right)  =\frac{S_{\theta}\left(
F\right)  }{s_{\theta}\left(  F\right)  }$.
\end{enumerate}

\bigskip\noindent\textbf{Arbitrary case}. Now we generalize the above ideas to
arbitrary cases.

\begin{enumerate}
\item We treat $\alpha_{1},\ldots,\alpha_{d_{0}}$ as distinct indeterminates.

\item Without loss of generality, we index the roots as follows.
\[%
\begin{array}
[c]{ccrr}%
F_{0} & = & a_{0d_{0}}\prod\limits_{k=1}^{d_{0}}\left(  x-\alpha_{k}\right)
& \\
G_{i} & = & \prod\limits_{k=d_{0}-e_{i}+1}^{d_{0}}\left(  x-\alpha_{k}\right)
\  & \text{for }i=0,\ldots,t
\end{array}
\]
where $a_{0.d_{0}}\neq0$. Then%
\begin{equation}
F_{i}(\alpha_{j})=0\ \ \text{for }j=d_{0}-e_{i}+1,\ldots,d_{0}
\label{eqs:Fj_alphai}%
\end{equation}

\item From the definition of subresultant (Definition \ref{def:roots}), we
have \
\[
S_{\theta}(F)=a_{0d_{0}}^{\theta_{0}}\cdot\frac{\det\left[
\begin{array}
[c]{rrr|c}%
\alpha_{1}^0 F_{1}\left(  \alpha_{1}\right)   & \cdots &
\alpha_{d_0}^0 F_{1} \left(  \alpha_{d_{0}}\right)   & \\
\vdots~~~~~~~~ &  & \vdots~~~~~~~~ & \\
\alpha_{1}^{\theta_{1}-1} F_{1}(\alpha_{1}) & \cdots &
\alpha_{d_0}^{\theta_{1}-1} F_{1}(\alpha_{d_{0}}) & \\\hline
\vdots~~~~~~~~ &  & \vdots~~~~~~~~ & \\
\vdots~~~~~~~~ &  & \vdots~~~~~~~~ & \\\hline
\alpha_{1}^0 F_{t}\left(  \alpha_{1}\right)   & \cdots &
\alpha_{d_0}^0 F_{t} \left(  \alpha_{d_{0}}\right)   & \\
\vdots~~~~~~~~ &  & \vdots~~~~~~~~ & \\
\alpha_{1}^{\theta_{t}-1} F_{t}(\alpha_{1}) & \cdots &
\alpha_{d_0}^{\theta_{t}-1} F_{t}(\alpha_{d_{0}}) & \\\hline
\alpha_{1}^0~~~~~~~   & \cdots & \alpha_{d_{0}}^0~~~~~~~  & x^0\\
\vdots~~~~~~~~ &  & \vdots~~~~~~~~ & \vdots\\
\alpha_{1}^{\varepsilon-1}~~~~~~   & \cdots &\alpha_{d_{0}}^{\varepsilon-1}~~~~~~   & x^{\varepsilon-1}%
\end{array}
\right]  }{V\left(  \alpha_{1},\ldots,\alpha_{d_{0}}\right)  }\]
where $\varepsilon=1+d_{0}-\left\vert \theta\right\vert$ and $\theta_0=\max(d_1+\theta_1-d_0,\ldots,d_t+\theta_t-d_0,1-|\theta|)$.

\item Substitution of (\ref{eqs:Fj_alphai}) into the above expression yields
the following determinant with a block lower-triangular structure:%
\[
S_{\theta}(F)=a_{0d_{0}}^{\theta_{0}}\cdot\frac{\det\left[
\begin{array}
[c]{cccc}%
M_{1} &  &  & \\
\vdots & \ddots &  & \\
\cdot & \cdots & M_{t} & \\
\cdot & \cdots & \cdot & N
\end{array}
\right]  }{V(\alpha_{1},\ldots,\alpha_{d_{0}})}%
\]

where%
\[%
\begin{array}
[c]{lll}%
M_{i} & = & \left[
\begin{array}
[c]{ccc}%
\alpha_{d_{0}-e_{i-1}+1}^{0}F_{i}(\alpha_{d_{0}-e_{i-1}+1}), & \ldots &
\alpha_{d_{0}-e_{i}}^{0}F_{i}(\alpha_{d_{0}-e_{i}})\\
\vdots &  & \vdots\\
\alpha_{d_{0}-e_{i-1}+1}^{\theta_{i}-1}F_{i}(\alpha_{d_{0}-e_{i-1}+1}), &
\ldots & \alpha_{d_{0}-e_{i}}^{\theta_{i}-1}F_{i}(\alpha_{d_{0}-e_{i}})
\end{array}
\right]  \ \ \ \ \ \ \ \,\text{for }i=1,\ldots,t\\
\  &  & \\
N & = & \left[
\begin{array}
[c]{cccc}%
\alpha_{d_{0}-e_{t}+1}^{0} & \ldots & \alpha_{d_{0}}^{0} & x^{0}\\
\vdots &  & \vdots & \vdots\\
\alpha_{d_{0}-e_{t}+1}^{\varepsilon-1} & \ldots & \alpha_{d_{0}}%
^{\varepsilon-1} & x^{\varepsilon-1}%
\end{array}
\right]
\end{array}
\]

\item Note that $M_{1},\ldots,M_{t},N$ are all square matrices since%
\begin{align*}
\#\ \text{of\ rows of }M_{i}  &  =\theta_{i}\\
&  =e_{i-1}-e_{i}\\
&  =\ \text{\# of columns\ \thinspace of }M_{i}\\
\#\ \text{of rows of }N\  &  =\varepsilon\\
&  =1+d_{0}-\left(  \theta_{1}+\cdots+\theta_{t}\right) \\
&  =1+d_{0}-\left(  e_{0}-e_{1}+e_{1}-e_{2}+\cdots+e_{t-1}-e_{t}\right) \\
&  =1+d_{0}-\left(  e_{0}-e_{t}\right) \\
&  =1+e_{t}\ \ \ \ \ \ \text{since }d_{0}=e_{0}\text{ }\\
&  =\ \text{\# of columns of }N
\end{align*}

\item By applying elementary properties of determinants to the above
expressions, we have%
\[
{S}_{\theta}(F)=a_{0d_{0}}^{\theta_{0}}\cdot\frac{\prod\limits_{i=1}^{t}\det M_{i}\ \ \det N}{V\left(
\alpha_{1},\ldots,\alpha_{d_{0}}\right)  }%
\]
where%
\[%
\begin{array}
[c]{lll}%
\det M_{i} & = & V(\alpha_{d_{0}-e_{i-1}+1},\ldots,\alpha_{d_{0}-e_{i}}%
)\prod\limits_{j=d_{0}-e_{i-1}+1}^{d_{0}-e_{i}}F_{i}(\alpha_{j}%
)\ \ \ \ \ \ \text{for }i=1,\ldots,t\\
\det N & = & V(\alpha_{d_{0}-e_{t}+1},\ldots,\alpha_{d_{0}})\ \gcd\left(
F\right)  \ \ \ \ \ \ \text{since }G_{t}=\gcd\left(  F\right)
\end{array}
\]

\item Thus we have%
\[
S_{\theta}\left(  F\right)  =s_{\theta}\left(  F\right)  \ \gcd\left(
F\right)
\]
where%
\[
s_{\theta}\left(  F\right)  =a_{0d_{0}}^{\theta_{0}}\cdot\frac{\prod\limits_{i=1}^{t+1}V(\alpha
_{d_{0}-e_{i-1}+1},\ldots,\alpha_{d_{0}-e_{i}})\ }{V(\alpha_{1},\ldots
,\alpha_{d_{0}})}\ \prod\limits_{i=1}^{t}\prod\limits_{j=d_{0}-e_{i-1}%
+1}^{d_{0}-e_{i}}F_{i}(\alpha_{j})
\]
where $e_{t+1}=0$.

\item By applying elementary properties of Vandermonde matrices, we have%
\begin{align*}
&  \frac{\prod\limits_{i=1}^{t+1}V(\alpha_{d_{0}-e_{i-1}+1},\ldots
,\alpha_{d_{0}-e_{i}})}{V(\alpha_{1},\ldots,\alpha_{d_{0}})}\\
&  =\frac{1}{\prod\limits_{i=1}^{t}\prod\limits_{j=d_{0}-e_{i-1}+1}%
^{d_{0}-e_{i}}\prod\limits_{k=d_{0}-e_{i}+1}^{d_{0}}\left(  \alpha_{k}%
-\alpha_{j}\right)  }\\
&  =\pm\frac{1}{\prod\limits_{i=1}^{t}\prod\limits_{j=d_{0}-e_{i-1}+1}%
^{d_{0}-e_{i}}\prod\limits_{k=d_{0}-e_{i}+1}^{d_{0}}\left(  \alpha_{j}%
-\alpha_{k}\right)  }\\
&  =\pm\frac{1}{\prod\limits_{i=1}^{t}\prod\limits_{j=d_{0}-e_{i-1}+1}%
^{d_{0}-e_{i}}G\left(  \alpha_{j}\right)  }%
\end{align*}

\item Thus
\[
s_{\theta}(F)=\pm a_{0d_{0}}^{\theta_{0}}\cdot\prod\limits_{i=1}^{t}\prod\limits_{j=d_{0}-e_{i-1}+1}%
^{d_{0}-e_{i}}\frac{F_{i}\left(  \alpha_{j}\right)  }{G_{i}\left(  \alpha
_{j}\right)  }=\pm a_{0d_{0}}^{\theta_{0}}\cdot\prod\limits_{i=1}^{t}\prod\limits_{j=d_{0}-e_{i-1}%
+1}^{d_{0}-e_{i}}H_{i}(\alpha_{j})
\]
\newline where $H_{i}=F_{i}/G_{i}$.

\item From now on, we treat $\alpha_{1},\ldots,\alpha_{d_{0}}$ as numbers.

\item Note that $s_{\theta}\left(  F\right)  =\pm a_{0d_{0}}^{\theta_{0}}\cdot\prod\limits_{i=1}^{t}%
\prod\limits_{j=d_{0}-e_{i-1}+1}^{d_{0}-e_{i}}H_{i}(\alpha_{j})\neq0$.

\item Thus we also have $\gcd(F)=\frac{S_{\theta}\left(  F\right)  }%
{s_{\theta}\left(  F\right)  }$.
\end{enumerate}
\end{proof}

\begin{lemma}
\label{lem:part_b}Let $\theta=\operatorname*{icdeg}\left(  F\right)  $. Let
$\delta\in\mathbb{N}_{\geq0}^{t}$\ such that $\left\vert \delta\right\vert
\leq d_{0}$. \ If $\delta\succ_{\text{glex}}\theta$ then $S_{\delta}=0$.
\end{lemma}

\begin{proof}
\ In order to convey the main underlying ideas effectively, we will show the
proof for a particular case first. After that, we will generalize the ideas to
arbitrary cases.

\bigskip

\noindent\textbf{Particular case}. We consider the same particular case used
in the proof of Lemma \ref{lem:part_a}. Recall that $d=\left(  7,6,6\right)
,$ $e=\left(  7,4,2\right)  $ and $\theta=\left(  3,2\right)  $.

\begin{enumerate}
\item We treat $\alpha_{1},\ldots,\alpha_{7}$ as distinct indeterminates.

\item Let $\delta=\left(  3,3\right)  $. Note $\delta\succ_{\text{glex}}%
\theta$. Note that $\delta_{1}+\delta_{2}>\theta_{1}+\theta_{2}$.

\item From the definition of subresultant (Definition \ref{def:roots}) we
have
\[
S_{\delta}\left(  F\right)  =a_{07}^{2}\cdot\frac{\det\left[
\begin{array}
[c]{ccccccc|c}%
\alpha_{1}^{0}F_{1}\left(  \alpha_{1}\right)  & \alpha_{2}^{0}F_{1}\left(
\alpha_{2}\right)  & \alpha_{3}^{0}F_{1}\left(  \alpha_{3}\right)  &
\alpha_{4}^{0}F_{1}\left(  \alpha_{4}\right)  & \alpha_{5}^{0}F_{1}\left(
\alpha_{5}\right)  & \alpha_{6}^{0}F_{1}\left(  \alpha_{6}\right)  &
\alpha_{7}^{0}F_{1}\left(  \alpha_{7}\right)  & \\
\alpha_{1}^{1}F_{1}\left(  \alpha_{1}\right)  & \alpha_{2}^{1}F_{1}\left(
\alpha_{2}\right)  & \alpha_{3}^{1}F_{1}\left(  \alpha_{3}\right)  &
\alpha_{4}^{1}F_{1}\left(  \alpha_{4}\right)  & \alpha_{5}^{1}F_{1}\left(
\alpha_{5}\right)  & \alpha_{6}^{1}F_{1}\left(  \alpha_{6}\right)  &
\alpha_{7}^{1}F_{1}\left(  \alpha_{7}\right)  & \\
\alpha_{1}^{2}F_{1}(\alpha_{1}) & \alpha_{2}^{2}F_{1}(\alpha_{2}) & \alpha
_{3}^{2}F_{1}(\alpha_{3}) & \alpha_{4}^{2}F_{1}(\alpha_{4}) & \alpha_{5}%
^{2}F_{1}(\alpha_{5}) & \alpha_{6}^{2}F_{1}(\alpha_{6}) & \alpha_{7}^{2}%
F_{1}(\alpha_{7}) & \\\hline
\alpha_{1}^{0}F_{2}\left(  \alpha_{1}\right)  & \alpha_{2}^{0}F_{2}\left(
\alpha_{2}\right)  & \alpha_{3}^{0}F_{2}\left(  \alpha_{3}\right)  &
\alpha_{4}^{0}F_{2}\left(  \alpha_{4}\right)  & \alpha_{5}^{0}F_{2}\left(
\alpha_{5}\right)  & \alpha_{6}^{0}F_{2}\left(  \alpha_{6}\right)  &
\alpha_{7}^{0}F_{2}\left(  \alpha_{7}\right)  & \\
\alpha_{1}^{1}F_{2}\left(  \alpha_{1}\right)  & \alpha_{2}^{1}F_{2}\left(
\alpha_{2}\right)  & \alpha_{3}^{1}F_{2}\left(  \alpha_{3}\right)  &
\alpha_{4}^{1}F_{2}\left(  \alpha_{4}\right)  & \alpha_{5}^{1}F_{2}\left(
\alpha_{5}\right)  & \alpha_{6}^{1}F_{2}\left(  \alpha_{6}\right)  &
\alpha_{7}^{1}F_{2}\left(  \alpha_{7}\right)  & \\
\alpha_{1}^{2}F_{2}(\alpha_{1}) & \alpha_{2}^{2}F_{2}(\alpha_{2}) & \alpha
_{3}^{2}F_{2}(\alpha_{3}) & \alpha_{4}^{2}F_{2}(\alpha_{4}) & \alpha_{5}%
^{2}F_{2}(\alpha_{5}) & \alpha_{6}^{2}F_{2}(\alpha_{6}) & \alpha_{7}^{2}%
F_{2}(\alpha_{7}) & \\\hline
\alpha_{1}^{0} & \alpha_{2}^{0} & \alpha_{3}^{0} & \alpha_{4}^{0} & \alpha
_{5}^{0} & \alpha_{6}^{0} & \alpha_{7}^{0} & x^{0}\\
\alpha_{1}^{1} & \alpha_{2}^{1} & \alpha_{3}^{1} & \alpha_{4}^{1} & \alpha
_{5}^{1} & \alpha_{6}^{1} & \alpha_{7}^{1} & x^{1}%
\end{array}
\right]  }{V\left(  \alpha_{1},\ldots,\alpha_{7}\right)  }%
\]
where
\begin{align*}
\varepsilon&=1+d_{0}-\left\vert \delta\right\vert =1+7-(3+3)=2\\
\delta_0&=\max(d_1+\delta_1-d_0,d_2+\delta_2-d_0,1-|\delta|)\\
&=\max(6+3-7,6+3-7,1-(3+3))=2
\end{align*}

\item Substitution of (\ref{eq:Fj_alphai_ex}) into the above expression yields
the following determinant with a block lower-triangular structure:%
\begin{align*}
S_{\delta}\left(  F\right)   &  =a_{07}^{2}\cdot\frac{\det\left[
\begin{array}
[c]{ccc|cc|ccc}%
\alpha_{1}^{0}F_{1}\left(  \alpha_{1}\right)  & \alpha_{2}^{0}F_{1}\left(
\alpha_{2}\right)  & \alpha_{3}^{0}F_{1}\left(  \alpha_{3}\right)  &  &  &  &
& \\
\alpha_{1}^{1}F_{1}\left(  \alpha_{1}\right)  & \alpha_{2}^{1}F_{1}\left(
\alpha_{2}\right)  & \alpha_{3}^{1}F_{1}\left(  \alpha_{3}\right)  &  &  &  &
& \\
\alpha_{1}^{2}F_{1}(\alpha_{1}) & \alpha_{2}^{2}F_{1}(\alpha_{2}) & \alpha
_{3}^{2}F_{1}(\alpha_{3}) &  &  &  &  & \\\hline
\alpha_{1}^{0}F_{2}\left(  \alpha_{1}\right)  & \alpha_{2}^{0}F_{2}\left(
\alpha_{2}\right)  & \alpha_{3}^{0}F_{2}\left(  \alpha_{3}\right)  &
\alpha_{4}^{0}F_{2}\left(  \alpha_{4}\right)  & \alpha_{5}^{0}F_{2}\left(
\alpha_{5}\right)  &  &  & \\
\alpha_{1}^{1}F_{2}\left(  \alpha_{1}\right)  & \alpha_{2}^{1}F_{2}\left(
\alpha_{2}\right)  & \alpha_{3}^{1}F_{2}\left(  \alpha_{3}\right)  &
\alpha_{4}^{1}F_{2}\left(  \alpha_{4}\right)  & \alpha_{5}^{1}F_{2}\left(
\alpha_{5}\right)  &  &  & \\
\alpha_{1}^{2}F_{2}(\alpha_{1}) & \alpha_{2}^{2}F_{2}(\alpha_{2}) & \alpha
_{3}^{2}F_{2}(\alpha_{3}) & \alpha_{4}^{2}F_{2}(\alpha_{4}) & \alpha_{5}%
^{2}F_{2}(\alpha_{5}) &  &  & \\\hline
\alpha_{1}^{0} & \alpha_{2}^{0} & \alpha_{3}^{0} & \alpha_{4}^{0} & \alpha
_{5}^{0} & \alpha_{6}^{0} & \alpha_{7}^{0} & x^{0}\\
\alpha_{1}^{1} & \alpha_{2}^{1} & \alpha_{3}^{1} & \alpha_{4}^{1} & \alpha
_{5}^{1} & \alpha_{6}^{1} & \alpha_{7}^{1} & x^{1}%
\end{array}
\right]  }{V\left(  \alpha_{1},\ldots,\alpha_{7}\right)  }\\
&  =a_{07}^{2}\cdot\frac{\det\left[
\begin{array}
[c]{ccc}%
M_{1} &  & \\
\cdot & M_{2} & \\
\cdot & \cdot & N
\end{array}
\right]  }{V\left(  \alpha_{1},\ldots,\alpha_{7}\right)  }%
\end{align*}
where%
\begin{align*}
M_{1}  &  =\left[
\begin{array}
[c]{ccc}%
\alpha_{1}^{0}F_{1}\left(  \alpha_{1}\right)  & \alpha_{2}^{0}F_{1}\left(
\alpha_{2}\right)  & \alpha_{3}^{0}F_{1}\left(  \alpha_{3}\right) \\
\alpha_{1}^{1}F_{1}\left(  \alpha_{1}\right)  & \alpha_{2}^{1}F_{1}\left(
\alpha_{2}\right)  & \alpha_{3}^{1}F_{1}\left(  \alpha_{3}\right) \\
\alpha_{1}^{2}F_{1}(\alpha_{1}) & \alpha_{2}^{2}F_{1}(\alpha_{2}) & \alpha
_{3}^{2}F_{1}(\alpha_{3})
\end{array}
\right] \\
M_{2}  &  =\left[
\begin{array}
[c]{cc}%
\alpha_{4}^{0}F_{2}\left(  \alpha_{4}\right)  & \alpha_{5}^{0}F_{2}\left(
\alpha_{5}\right) \\
\alpha_{4}^{1}F_{2}\left(  \alpha_{4}\right)  & \alpha_{5}^{1}F_{2}\left(
\alpha_{5}\right) \\
\alpha_{4}^{2}F_{2}(\alpha_{4}) & \alpha_{5}^{2}F_{2}(\alpha_{5})
\end{array}
\right] \\
N  &  =\left[
\begin{array}
[c]{ccc}%
\alpha_{6}^{0} & \alpha_{7}^{0} & x^{0}\\
\alpha_{6}^{1} & \alpha_{7}^{1} & x^{1}%
\end{array}
\right]
\end{align*}

\item We repartition the numerator matrix so that the diagonal consists of two
square matrices
\begin{align*}
S_{\delta}\left(  F\right)   &  =a_{07}^{2}\cdot\frac{\det\left[
\begin{array}
[c]{cccccc|cc}%
\alpha_{1}^{0}F_{1}\left(  \alpha_{1}\right)  & \alpha_{2}^{0}F_{1}\left(
\alpha_{2}\right)  & \alpha_{3}^{0}F_{1}\left(  \alpha_{3}\right)  &  &  &  &
& \\
\alpha_{1}^{1}F_{1}\left(  \alpha_{1}\right)  & \alpha_{2}^{1}F_{1}\left(
\alpha_{2}\right)  & \alpha_{3}^{1}F_{1}\left(  \alpha_{3}\right)  &  &  &  &
& \\
\alpha_{1}^{2}F_{1}(\alpha_{1}) & \alpha_{2}^{2}F_{1}(\alpha_{2}) & \alpha
_{3}^{2}F_{1}(\alpha_{3}) &  &  &  &  & \\
\alpha_{1}^{0}F_{2}\left(  \alpha_{1}\right)  & \alpha_{2}^{0}F_{2}\left(
\alpha_{2}\right)  & \alpha_{3}^{0}F_{2}\left(  \alpha_{3}\right)  &
\alpha_{4}^{0}F_{2}\left(  \alpha_{4}\right)  & \alpha_{5}^{0}F_{2}\left(
\alpha_{5}\right)  &  &  & \\
\alpha_{1}^{1}F_{2}\left(  \alpha_{1}\right)  & \alpha_{2}^{1}F_{2}\left(
\alpha_{2}\right)  & \alpha_{3}^{1}F_{2}\left(  \alpha_{3}\right)  &
\alpha_{4}^{1}F_{2}\left(  \alpha_{4}\right)  & \alpha_{5}^{1}F_{2}\left(
\alpha_{5}\right)  &  &  & \\
\alpha_{1}^{2}F_{2}(\alpha_{1}) & \alpha_{2}^{2}F_{2}(\alpha_{2}) & \alpha
_{3}^{2}F_{2}(\alpha_{3}) & \alpha_{4}^{2}F_{2}(\alpha_{4}) & \alpha_{5}%
^{2}F_{2}(\alpha_{5}) &  &  & \\\hline
\alpha_{1}^{0} & \alpha_{2}^{0} & \alpha_{3}^{0} & \alpha_{4}^{0} & \alpha
_{5}^{0} & \alpha_{6}^{0} & \alpha_{7}^{0} & x^{0}\\
\alpha_{1}^{1} & \alpha_{2}^{1} & \alpha_{3}^{1} & \alpha_{4}^{1} & \alpha
_{5}^{1} & \alpha_{6}^{1} & \alpha_{7}^{1} & x^{1}%
\end{array}
\right]  }{V\left(  \alpha_{1},\ldots,\alpha_{7}\right)  }\\
&  =a_{07}^{2}\cdot\frac{\det\left[
\begin{array}
[c]{cc}%
T & \\
\cdot & B
\end{array}
\right]  }{V\left(  \alpha_{1},\ldots,\alpha_{7}\right)  }%
\end{align*}

where the size of the square matrix $T$ is $\delta_{1}+\delta_{2}=3+3=6$,
namely%
\[
T=\left[
\begin{array}
[c]{ccc}%
M_{1} &  & \\
\cdot & M_{2} & 0
\end{array}
\right]
\]
where $0$ is the $\delta_{2}\times p$ matrix with zeros, where again
$\delta_{2}=3\ $and $p=\left(  \delta_{1}+\delta_{2}\right)  -\left(
\theta_{1}+\theta_{2}\right)  =1$.

\item By applying elementary properties of determinants to the above
expressions, we have%
\[
{S}_{\theta}(F)=a_{07}^{2}\cdot\frac{\det T\det B}{V\left(  \alpha_{1},\ldots,\alpha
_{7}\right)  }%
\]

\item Since $p=1>0$, the last column of $T$ is all zero. Hence $\det T=0$ and
in turn ${S}_{\theta}(F)=0$.

\item From now on, we treat $\alpha_{1},\ldots,\alpha_{7}$ as numbers.

\item Obviously ${S}_{\theta}(F)=0$.
\end{enumerate}

\bigskip\noindent\textbf{Arbitrary case}. Now we generalize the above ideas to
arbitrary cases. Let $\delta\in\mathbb{N}_{\geq0}^{t}\ $be such that
$\left\vert \delta\right\vert \leq d_{0}$.

\begin{enumerate}
\item We will treat $\alpha_{1},\ldots,\alpha_{d_{0}}$ as distinct indeterminates.

\item Assume $\delta\succ_{\text{glex}}\theta$. Then for some $\ell$, we have
$\delta_{1}+\cdots+\delta_{\ell}>\theta_{1}+\cdots+\theta_{\ell}$.

\item Recall that
\[
S_{\delta}(F)=a_{0d_{0}}^{\delta_{0}}\cdot\dfrac{\det\left[
\begin{array}
[c]{rrr|c}%
\alpha_{1}^0 F_{1}\left(  \alpha_{1}\right)   & \cdots &
\alpha_{d_0}^0 F_{1} \left(  \alpha_{d_{0}}\right)   & \\
\vdots~~~~~~~~ &  & \vdots~~~~~~~~ & \\
\alpha_{1}^{\delta_{1}-1} F_{1}(\alpha_{1}) & \cdots &
\alpha_{d_0}^{\delta_{1}-1} F_{1}(\alpha_{d_{0}}) & \\\hline
\vdots~~~~~~~~ &  & \vdots~~~~~~~~ & \\
\vdots~~~~~~~~ &  & \vdots~~~~~~~~ & \\\hline
\alpha_{1}^0 F_{t}\left(  \alpha_{1}\right)   & \cdots &
\alpha_{d_0}^0 F_{t} \left(  \alpha_{d_{0}}\right)   & \\
\vdots~~~~~~~~ &  & \vdots~~~~~~~~ & \\
\alpha_{1}^{\delta_{t}-1} F_{t}(\alpha_{1}) & \cdots &
\alpha_{d_0}^{\delta_{t}-1} F_{t}(\alpha_{d_{0}}) & \\\hline
\alpha_{1}^0~~~~~~~   & \cdots & \alpha_{d_{0}}^0~~~~~~~  & x^0\\
\vdots~~~~~~~~ &  & \vdots~~~~~~~~ & \vdots\\
\alpha_{1}^{\varepsilon-1}~~~~~~   & \cdots &\alpha_{d_{0}}^{\varepsilon-1}~~~~~~   & x^{\varepsilon-1}%
\end{array}
\right]}
{V\left(  \alpha_{1},\ldots,\alpha_{d_{0}}\right)}\]
where $a_{0d_0}$ is the leading coefficient of $F_0$, and
\begin{align*}
\varepsilon&=1+d_{0}-\left\vert \delta\right\vert \\
\delta_0&=\max(d_1+\delta_1-d_0,\ldots,d_t+\delta_t-d_0,1-|\delta|)
\end{align*}

\item Substitution of \eqref{eqs:Fj_alphai} into the above expression yields
the following determinant with a block lower-triangular structure:
\begin{equation}
S_{\delta}(F)=a_{0d_{0}}^{\delta_{0}}\cdot\frac{\det\left[
\begin{array}
[c]{cccc}%
M_{1} &  &  & \\
\vdots & \ddots &  & \\
\cdot & \cdots & M_{t} & \\
\cdot & \cdots & \cdot & N
\end{array}
\right]  }{V(\alpha_{1},\ldots,\alpha_{d_{0}})} \label{eq:lowtriblock}%
\end{equation}
where $M_{i}$ is $\delta_{i}$ by $\theta_{i}$.

\item We repartition the numerator matrix so that the diagonal consists of two
square matrices $T$ and $B$ as follows%
\[
S_{\delta}\left(  F\right)  =a_{0d_{0}}^{\delta_{0}}\cdot\frac{\det\left[
\begin{array}
[c]{cc}%
T & \\
\cdot & B
\end{array}
\right]  }{V\left(  \alpha_{1},\ldots,\alpha_{d_{0}}\right)  }%
\]

where the size of the square matrix $T$ is $\delta_{1}+\cdots+\delta_{\ell}$,
namely%
\[
T=\left[
\begin{array}
[c]{cccc}%
M_{1} &  &  & \\
\vdots & \ddots &  & \\
\cdot & \cdots & M_{\ell} & 0
\end{array}
\right]
\]
where $0$ is the $\delta_{\ell}\times p$ matrix with zeros, where again
$p=\left(  \delta_{1}+\cdots+\delta_{\ell}\right)  -\left(  \theta_{1}%
+\cdots+\theta_{\ell}\right)  $.

\item By applying elementary properties of determinants to the above
expressions, we have%
\[
{S}_{\theta}(F)=a_{0d_{0}}^{\delta_{0}}\cdot\frac{\det T\det B}{V\left(  \alpha_{1},\ldots,\alpha_{d_{0}%
}\right)  }%
\]

\item Since $p>0$, the last column of $T$ is all zero. Hence $\det T=0$ and in
turn ${S}_{\theta}(F)=0$.

\item From now on, we treat $\alpha_{1},\ldots,\alpha_{d_{0}}$ as numbers.

\item Obviously ${S}_{\theta}(F)=0$.
\end{enumerate}
\end{proof}

\begin{proof}
[Proof of Theorem \ref{thm:gcd}].

\begin{enumerate}
\item Let $\delta=\max\limits_{s_{\gamma}\left(  F\right)  \neq0}\gamma$ and
let $\theta=\operatorname*{icdeg}\left(  F\right)  $.

\item Obviously $s_{\delta}\left(  F\right)  \neq0$.

\item Hence from the contra-positive of Lemma \ref{lem:part_b}, we have
$\delta\preceq_{\text{glex}}\theta$.

\item From Lemma \ref{lem:part_a}, we have $s_{\theta}\left(  F\right)  \neq
0$. Recalling $\delta=\max\limits_{s_{\gamma}\left(  F\right)  \neq0}\gamma$,
we have $\delta\succeq_{\text{glex}}\theta$.

\item Thus $\theta=\delta$, that is, $\operatorname*{icdeg}\left(  F\right)
=\delta$.

\item By Lemma \ref{lem:part_a}, we have $\gcd\left(  F\right)  =\frac
{S_{\delta}}{s_{\delta}}$.
\end{enumerate}
\end{proof}

\section{Application: Parametric multiplicity}

\label{sec:application_parametric_multiplicity} In this section, we show an
application of the subresultant polynomial in computing the multiplicity of a
parametric univariate polynomial.

\subsection{Main results}

\begin{definition}
[Multiplicity]Let $H\in\mathbb{C}\left[  x\right]  $ with $m$ distinct complex
roots, say $r_{1},\ldots,r_{m}$ with multiplicities $\mu_{1},\ldots,\mu_{m}$
respectively. Without losing generality, we assume that $\mu_{1}\geq\cdots
\geq\mu_{m}$. Then the \emph{multiplicity} of~$H$, written as
$\operatorname*{mult}\left(  F\right)  $, is defined by%
\[
\operatorname*{mult}\left(  H\right)  =\left(  \mu_{1},\ldots,\mu_{m}\right)
\]

\end{definition}

\begin{definition}
[Conjugate]Let $\delta=\left(  \delta_{1},\ldots,\delta_{t}\right)  $. Then
the \emph{conjugate} $\bar{\delta}=\left(  \bar{\delta}_{1},\ldots,\bar
{\delta}_{s}\right)  $ of $\delta$ is defined by%
\begin{align*}
s  &  =\max\delta\\
\bar{\delta}_{i}  &  =\#\left\{  j\in\left[  1,\ldots,t\right]  :\delta
_{j}\geq i\right\}  \ \ \ \text{for }i=1,\ldots,s
\end{align*}

\end{definition}

\begin{example}
\

\begin{enumerate}
\item Let $\delta=\left(  3,2,0,0,0\right)  $. Note%
\begin{align*}
s  &  =\max\delta=3\\
\bar{\delta}_{1}  &  =\#\left\{  j:\delta_{j}\geq1\right\}  =2\\
\bar{\delta}_{2}  &  =\#\left\{  j:\delta_{j}\geq2\right\}  =2\\
\bar{\delta}_{3}  &  =\#\left\{  j:\delta_{j}\geq2\right\}  =1
\end{align*}
Thus%
\[
\bar{\delta}=\left(  2,2,1\right)
\]

\end{enumerate}
\end{example}

\begin{theorem}
[Parametric multiplicity in terms of subresultants]\label{thm:multiplicity}Let
$H\in\mathbb{C}\left[  x\right]  $ be of degree $t$ and let
\[
\delta=\max\limits_{\substack{s_{\lambda}\left(  F\right)  \neq0}}\lambda
\]
where

\begin{enumerate}
\item $F=\left(  H^{\left(  0\right)  },H^{\left(  1\right)  },\ldots
,H^{\left(  t\right)  }\right)  $ and $H^{(i)}$ is the $i$-th derivative of $H$;

\item $\lambda=\left(  \lambda_{1},\ldots,\lambda_{t}\right)  $ such that
$\lambda_{1}+\cdots+\lambda_{t}=t\ $and $\lambda_{1}\geq\lambda_{2}\cdots
\geq\lambda_{t}\geq0$;

\item $\max$ is with respect to the lexicographic ordering $\succ_{\text{lex}%
}$.
\end{enumerate}

\noindent Then we have%
\[
\operatorname*{mult}(H)=\bar{\delta}%
\]

\end{theorem}

\begin{remark}
\ In the above, $\max$ always exists since ${s}_{\left(  1,\ldots,1\right)
}\left(  F\right)  \neq0$ from Definition~\ref{def:roots}.
\end{remark}


\begin{example}
We have the following parametric multiplicity for degree $5$.
\[%
\begin{array}
[c]{llclll}%
\text{if} & {s}_{\left(  5,0,0,0,0\right)  }\left(  F\right)  \neq0 &  &
\text{then} &  & \operatorname*{mult}(H)=(1,1,1,1,1)\\
\text{else if} & {s}_{\left(  4,1,0,0,0\right)  }\left(  F\right)  \neq0 &  &
\text{then} &  & \operatorname*{mult}(H)=(2,1,1,1)\\
\text{else if} & {s}_{\left(  3,2,0,0,0\right)  }\left(  F\right)  \neq0 &  &
\text{then} &  & \operatorname*{mult}(H)=(2,2,1)\\
\text{else if} & {s}_{\left(  3,1,1,0,0\right)  }\left(  F\right)  \neq0 &  &
\text{then} &  & \operatorname*{mult}(H)=(3,1,1)\\
\text{else if} & {s}_{\left(  2,2,1,0,0\right)  }\left(  F\right)  \neq0 &  &
\text{then} &  & \operatorname*{mult}(H)=(3,2)\\
\text{else if} & {s}_{\left(  2,1,1,1,0\right)  }\left(  F\right)  \neq0 &  &
\text{then} &  & \operatorname*{mult}(H)=\left(  4,1\right) \\
\text{else} &  &  & \text{then} &  & \operatorname*{mult}(H)=\left(  5\right)
\end{array}
\]

\end{example}

\begin{remark}
Previous works for generating the condition for complex root classification
are mostly based on repeated gcd computation using subresultants for two
polynomials \cite{1998_Gonzalez_Recio_Lombardi,1996_Yang_Hou_Zeng}%
\footnote[2]{The authors considered a more general problem, i.e., the real
root classification problem. When restricted to the complex case, the methods
gave conditions for complex root classification.}
We compare them with the conditions given in this paper (Theorem~\ref{thm:multiplicity}).
\begin{itemize}
\item In the previous works, the polynomials in the conditions are principal coefficients of nested subresultant polynomials,
that is, subresultant of subresultant polynomials of .....  and so on.
In the current work, the polynomials are just principal multi-polynomial subresultants of the input.

\item In the previous works, the number of polynomials in the conditions are often bigger than the number of multiplicity structures.
In the current work, the number of polynomials is exactly one less than that of
multiplicity structures.
\end{itemize}
\end{remark}

\subsection{Proof of Theorem \ref{thm:multiplicity} (Parametric Multiplicity)}

\begin{lemma}
\label{lem:icdeg}Let $F=\left(  H^{\left(  0\right)  },H^{\left(  1\right)
},\ldots,H^{\left(  t\right)  }\right)  $ and let $\delta
=\operatorname*{icdeg}\left(  F\right)  $. Then $\operatorname*{mult}%
(H)=\bar{\delta}$.
\end{lemma}

\begin{proof}
\ In order to convey the main underlying ideas effectively, we will show the
proof for a particular case first. After that, we will generalize the ideas to
arbitrary cases.

\bigskip

\noindent\textbf{Particular case}. Let $H=a\left(  x-r_{1}\right)  ^{2}\left(
x-r_{2}\right)  ^{2}\left(  x-r_{3}\right)  ^{1}$ where $a\neq0\ $and
$r_{1},r_{2},r_{3}$ are distinct.

\begin{enumerate}
\item Note $\deg H=5\ \,\,$and $\operatorname*{mult}(H)=\left(  2,2,1\right)
$.

\item Let $G_{i}=\gcd\left(  H^{\left(  0\right)  },\ldots,H^{\left(
i\right)  }\right)  $ for $i=0,\ldots,5$. Let $C_{i}=\frac{G_{i-1}}{G_{i}}$
for $i=1,\ldots,5$.Then%
\[%
\begin{array}
[c]{lllllll}%
G_{0} & = & \left(  x-r_{1}\right)  ^{2}\left(  x-r_{2}\right)  ^{2}\left(
x-r_{3}\right)  ^{1} &  &  &  & \\
G_{1} & = & \left(  x-r_{1}\right)  ^{1}\left(  x-r_{2}\right)  ^{1} &  &
C_{1} & = & \left(  x-r_{1}\right)  \left(  x-r_{2}\right)  \left(
x-r_{3}\right) \\
G_{2} & = & 1 &  & C_{2} & = & \left(  x-r_{1}\right)  \left(  x-r_{2}\right)
\\
G_{3} & = & 1 &  & C_{3} & = & 1\\
G_{4} & = & 1 &  & C_{4} & = & 1\\
G_{5} & = & 1 &  & C_{5} & = & 1
\end{array}
\]
Thus%
\[
\delta=\operatorname*{icdeg}\left(  F\right)  =\left(  3,2,0,0,0\right)
\]

\item Note%
\begin{align*}
s  &  =\max\delta=3\\
\bar{\delta}_{1}  &  =\#\left\{  j\in\left[  1,\ldots,5\right]  :\delta
_{j}\geq1\right\}  =2\\
\bar{\delta}_{2}  &  =\#\left\{  j\in\left[  1,\ldots,5\right]  :\delta
_{j}\geq2\right\}  =2\\
\bar{\delta}_{3}  &  =\#\left\{  j\in\left[  1,\ldots,5\right]  :\delta
_{j}\geq2\right\}  =1
\end{align*}
Thus%
\[
\bar{\delta}=\left(  2,2,1\right)
\]

\item Thus%
\[
\operatorname*{mult}(H)=\left(  2,2,1\right)  =\bar{\delta}%
\]

\end{enumerate}

\noindent\textbf{Arbitrary case}. Let $H=a\left(  x-r_{1}\right)  ^{\mu_{1}%
}\cdots\left(  x-r_{m}\right)  ^{\mu_{m}}$ where $a\neq0\ $and $r_{1}%
,\ldots,r_{m}$ are distinct.

\begin{enumerate}
\item Note $\deg H=t=\mu_{1}+\cdots+\mu_{m}\ \,\,$and $\operatorname*{mult}%
(H)=\mu=\left(  \mu_{1},\ldots,\mu_{m}\right)  $.

\item Let $G_{i}=\gcd\left(  H^{\left(  0\right)  },\ldots,H^{\left(
i\right)  }\right)  $ for $i=0,\ldots,t$. Let $C_{i}=\frac{G_{i-1}}{G_{i}}$
for $i=1,\ldots,t$. Then%
\begin{align}
G_{i}  &  =\prod_{\substack{1\leq k\leq m\\\mu_{k}>i}}\left(  x-r_{k}\right)
^{\mu_{k}-i}\nonumber\\
C_{i}  &  =\frac{G_{i-1}}{G_{i}}=\frac{\prod\limits_{\substack{1\leq k\leq
m\\\mu_{k}>i-1}}\left(  x-r_{k}\right)  ^{\mu_{k}-\left(  i-1\right)  }}%
{\prod\limits_{\substack{1\leq k\leq m\\\mu_{k}>i}}\left(  x-r_{k}\right)
^{\mu_{k}-i}}=\frac{\prod\limits_{\substack{1\leq k\leq m\\\mu_{k}\geq
i}}\left(  x-r_{k}\right)  ^{\mu_{k}-\left(  i-1\right)  }}{\prod
\limits_{\substack{1\leq k\leq m\\\mu_{k}\geq i}}\left(  x-r_{k}\right)
^{\mu_{k}-i}}=\prod\limits_{\substack{1\leq k\leq m\\\mu_{k}\geq i}}\left(
x-r_{k}\right)  \label{eqs:Ci}%
\end{align}

Hence
\[
\operatorname*{icdeg}\left(  F\right)  =\delta=\left(  \delta_{1}%
,\ldots,\delta_{t}\right)  \ \,\text{where }\delta_{i}=\deg C_{i}=\#\left\{
k\in\left[  1,\ldots,m\right]  :\mu_{k}\geq i\right\}
\]

\item Note
\begin{align*}
s  &  =\max\delta=m\\
\bar{\delta}_{i}  &  =\#\left\{  j:\delta_{j}\geq i\right\}  =\#\left\{
j:\#\left\{  k\in\left[  1,\ldots,m\right]  :\mu_{k}\geq j\right\}  \geq
i\right\}  =\mu_{i}\ \ \ \,\text{for }i=1,\ldots,s
\end{align*}
Hence%
\[
\bar{\delta}=\mu
\]

\item Thus%
\[
\operatorname*{mult}(H)=\bar{\delta}%
\]

\end{enumerate}
\end{proof}


\begin{proof}
[Proof of Theorem \ref{thm:multiplicity}]\

\begin{enumerate}
\item Let $\delta^{\prime}$ be such that $\operatorname*{mult}(H)=\overline
{\delta^{\prime}}$.

Note that $\delta^{\prime}$ exists and it is unique.

\item By Lemma \ref{lem:icdeg}, we have $\operatorname*{icdeg}(F)=\delta
^{\prime}$.

\item From Lemma \ref{lem:part_a}, we have $s_{\delta^{\prime}}\neq0$.

\item By hypothesis, $\delta^{\prime}\preceq_{\text{lex}}\delta$.

\item Since $\operatorname*{icdeg}(F)=\delta^{\prime}$, by Lemma
\ref{lem:part_b}, $s_{\gamma}=0$ for any $\gamma$ satisfying $|\gamma|=t$ and
$\gamma\succ_{\text{lex}}\delta^{\prime}$.

\item Since $s_{\delta}\neq0$, $\delta\preceq_{\text{lex}}\delta^{\prime
}$.

\item Combining 4 and 6, we have $\delta^{\prime}=\delta$.

\item Therefore, $\operatorname*{mult}(H)=\bar{\delta}$.
\end{enumerate}
\end{proof}

\section{Multi-polynomial subresultants in terms of coefficients}

\label{sec:expression_in_coefficients} A natural question is how to express
$S_{\delta}$ in terms of coefficients. In this section, we give three
different determinantal expressions of $S_{\delta}$ in terms of coefficients
whose explicit forms are presented in Subsection
\ref{subsec:expression_in_coeffs_main}. The three expressions are extensions
of Sylvester-type subresultant, Barnett-type subresultant and B\'{e}zout
subresultant, respectively. The proof ideas and techniques are elaborated in
the remaining three subsections.

\subsection{Main results}

\label{subsec:expression_in_coeffs_main}

We begin by recalling a few basic notions/notations.

\begin{notation}
\label{notation:matrices}\

\begin{enumerate}
\item The \emph{companion matrix} of $A=a_{n}x^{n}+\cdots+a_{0}$ is given by
$\left[
\begin{array}
[c]{cccc}%
0 &  &  & -a_{0}/a_{n}\\
1 &  &  & -a_{1}/a_{n}\\
& \ddots &  & \vdots\\
&  & 1 & -a_{n-1}/a_{n}%
\end{array}
\right]  $

\item The \emph{B\'{e}zout matrix }$M$ of $A,B\in\mathbb{C}\left[  x\right]  $
is such that $\dfrac{\left\vert
\begin{array}
[c]{cc}%
A(x) & A(y)\\
B(x) & B(y)
\end{array}
\right\vert }{x-y}=[%
\begin{array}
[c]{ccc}%
y^{0} & \cdots & y^{\ell-1}%
\end{array}
]M\left[
\begin{array}
[c]{c}%
x^{\ell-1}\\
\vdots\\
x^{0}%
\end{array}
\right]  $ where $\ \ell=\max(\deg A,\deg B)$.

\item $\ X_{\delta,h}=%
\begin{array}
[c]{l}%
\begin{bmatrix}
x &  &  & \\
-1 & \ddots &  & \\
& \ddots & \ddots & \\
&  & \ddots & x\\
&  &  & -1\\
&  &  & \\
&  &  &
\end{bmatrix}
\left.
\begin{array}
[c]{c}%
\;\\
\;\\
\;\\
\;\\
\;\\
\;\\
\;\\
\
\end{array}
\right\}  \ h~\text{rows}\\
\underbrace{\hspace{8em}}_{d_{0}-\left(  \delta_{1}+\cdots+\delta_{t}\right)
~\text{columns}~}%
\end{array}
$
\end{enumerate}
\end{notation}

To give the expression of $S_{\delta}$ in coefficients, we first extend the
three well known subresultant matrices for two polynomials to those for
arbitrary number of polynomials.

\begin{definition}
[Extended subresultant matrices]\label{def:extended_matrix} The extended
subresultant matrices of Sylvester/Barnett/B\'{e}zout type are given by:

\begin{enumerate}
\item $M_{\delta}^{\text{Sylvester}}=\left[
\begin{array}
[c]{cccccccccc}%
R_{01} & \cdots & R_{0\delta_{0}} & \cdots & \cdots & \cdots & R_{t1} & \cdots
& R_{t\delta_{t}} & X_{\delta,d_{0}+\delta_{0}}%
\end{array}
\right]  ^{T}\in\mathbb{R}^{\left(  d_{0}+\delta_{0}\right)  \times\left(
d_{0}+\delta_{0}\right)  }$

where $R_{ij}\in\mathbb{R}^{\left(  d_{0}+\delta_{0}\right)  \times1}$such
that $R_{ij,k}=\operatorname*{coeff}\left(  x^{j-1}F_{i},x^{k-1}\right)
\ $if $\delta_{0}=\max(\left(  \delta_{1}%
+d_{1}\right)  -d_{0},\ldots,\left(\delta_{t}%
+d_{t}\right)  -d_{0},1-|\delta|)\ge0$.

\item $M_{\delta}^{\text{Barnett}}\ =\left[
\begin{array}
[c]{cccccccccc}%
R_{11} & \cdots & R_{1\delta_{1}} & \cdots & \cdots & \cdots & R_{t1} & \cdots
& R_{t\delta_{t}} & X_{\delta,d_{0}}%
\end{array}
\right]  ^{T}\in\mathbb{R}^{d_{0}\times d_{0}}$

where $R_{ij}\in\mathbb{R}^{d_{0}\times1}$ is the ${j}$-th column of
$F_{i}(C_{0})$ where $C_{0}$ is the companion matrix of $F_{0}$.

\item $M_{\delta}^{\text{B\'{e}zout}}\ \ =\left[
\begin{array}
[c]{cccccccccc}%
R_{11} & \cdots & R_{1\delta_{1}} & \cdots & \cdots & \cdots & R_{t1} & \cdots
& R_{t\delta_{t}} & X_{\delta,d_{0}}%
\end{array}
\right]  ^{T}\in\mathbb{R}^{d_{0}\times d_{0}}$

where $R_{ij}$ $\in\mathbb{R}^{d_{0}\times1}\ $is the ${j}$-th column of the
B\'{e}zout matrix of $F_{0}$ and $F_{i}$. Note that it is defined only when
$\deg F_{0}\geq\deg F_{i}$ for every $i$. Otherwise $R_{ij}$ may have
different size.
\end{enumerate}
\end{definition}

Then we have:

\begin{theorem}
[Subresultant polynomials in coefficients]\label{thm:coefficients} Let $\delta\neq\left(
0,\ldots,0\right)  $. Then

\begin{enumerate}
\item $S_{\delta}(F)=
\left\{
\begin{array}{ll}
c\det M_{\delta}^{\text{Sylvester}},&\text{if~}\delta_0\ge0,\\[8pt]
0,&\text{otherwise},
\end{array}
\right.$
where $c=(-1)^{d_{0}%
\delta_{0}}$.

\item $S_{\delta}(F)=c\det M_{\delta}^{\text{Barnett}}$ \ \ where $c=a_{0d_{0}}^{\delta_{0}}$.

\item $S_{\delta}(F)=c\det M_{\delta}^{\text{B\'{e}zout}}$ \ \ \ where
$c=\ a_{0d_{0}}^{\delta_0-|\delta|}$.
\end{enumerate}
\end{theorem}

\begin{remark}
When $\delta=\left(  0,\ldots,0\right)  $, Definition~\ref{def:roots}
immediately implies $S_{\delta}\left(  F\right)  =a_{0d_{0}}^{\delta_0-1}F_{0}$
where $\delta_0=\max(d_1-d_0,\ldots,d_t-d_0,1)$.
\end{remark}

\begin{example}
Let
\[
F=\left(  F_{0},F_{1},F_{2}\right)  =\left(  a_{04}x^{4}+a_{03}x^{3}%
+a_{02}x^{2}+a_{01}x+a_{00},a_{13}x^{3}+a_{12}x^{2}+a_{11}x+a_{10},a_{22}%
x^{2}+a_{21}x+a_{20}\right)
\]
and $\delta=(2,1)$. Then $$\delta_{0}=\max(\max(\delta_{1}+d_{1},\delta
_{2}+d_{2})-d_{0},0)=\max(\max(5,3)-4,0)=1$$
Thus we have
\begin{align*}
M_{\delta}^{\text{Sylvester}}  &  =\left[
\begin{array}
[c]{ccccc}%
a_{{00}} & a_{{01}} & a_{{02}} & a_{{03}} & a_{04}\\
a_{{10}} & a_{{11}} & a_{{12}} & a_{{13}} & 0\\
0 & a_{{10}} & a_{{11}} & a_{{12}} & a_{{13}}\\
a_{{20}} & a_{{21}} & a_{{22}} & 0 & 0\\
x & -1 & 0 & 0 & 0
\end{array}
\right] \\
M_{\delta}^{\text{Barnett}}  &  =\left[
\begin{array}
[c]{cccc}%
a_{{10}} & a_{{11}} & a_{{12}} & a_{{13}}\\
-{\dfrac{a_{{13}}a_{{00}}}{a_{{04}}}} & -{\dfrac{a_{{01}}a_{{13}}-a_{{04}%
}a_{{10}}}{a_{{04}}}} & -{\dfrac{a_{{02}}a_{{13}}-a_{{04}}a_{{11}}}{a_{{04}}}}
& -{\dfrac{a_{{03}}a_{{13}}-a_{{04}}a_{{12}}}{a_{{04}}}}\\
a_{{20}} & a_{{21}} & a_{{22}} & 0\\
x & -1 & 0 & 0
\end{array}
\right] \\
M_{\delta}^{\text{B\'{e}zout}}  &  =\left[
\begin{array}
[c]{cccc}%
a_{{04}}a_{{10}} & a_{{04}}a_{{11}} & a_{{04}}a_{{12}} & a_{{04}}a_{{13}}\\
-a_{{13}}a_{{00}}+a_{{03}}a_{{10}} & -a_{{01}}a_{{13}}+a_{{03}}a_{{11}%
}+a_{{04}}a_{{10}} & -a_{{02}}a_{{13}}+a_{{03}}a_{{12}}+a_{{04}}a_{{11}} &
a_{{04}}a_{{12}}\\
a_{{20}}a_{{04}} & a_{{21}}a_{{04}} & a_{{22}}a_{{04}} & 0\\
x & -1 & 0 & 0
\end{array}
\right]
\end{align*}
From the above three theorems, we have%
\[
S_{\delta}(F)=\det M_{\delta}^{\text{Sylvester}}=a_{{04}}^{1}\det M_{\delta
}^{\text{Barnett}}=a_{{04}}^{-2}\det M_{\delta}^{\text{B\'{e}zout}}%
\]

\end{example}

\subsection{Proof of Theorem \ref{thm:coefficients}-1 (Sylvester-type)}

\label{subsec:proof_Th6}

Here is a high level view of the proof. We start with converting $S_{\delta}$
into an equivalent expression which is easier to be connected with the
expression in coefficients. The equivalent expression will also be used for
proving the Barnett-type and B\'{e}zout-type expressions. Then
several proof techniques from
\cite{2007_DAndrea_Hong_Krick_Szanto,2021_Hong_Yang} are adapted. More
specifically, we first multiply the extended Sylvester matrix by the following
object:
\[
\left[
\begin{array}
[c]{cccccc}%
\alpha_{1}^{0} & \cdots & \alpha_{d_{0}}^{0} &  &  & \\
\vdots &  & \vdots &  &  & \\
\alpha_{1}^{d_{0}-1} &  & \alpha_{d_{0}}^{d_{0}-1} &  &  & \\
\alpha_{1}^{d_{0}} & \cdots & \alpha_{d_{0}}^{d_{0}} & 1 &  & \\
\vdots &  & \vdots &  & \ddots & \\
\alpha_{1}^{d_{0}+\delta_{0}-1} & \cdots & \alpha_{d_{0}}^{d_{0}+\delta_{0}-1}
&  &  & 1
\end{array}
\right]
\]
whose determinant is the same as $V(\alpha_{1},\ldots,\alpha_{d_{0}})$, the
denominator of ${S}_{\delta}$. Next we repeatedly rewrite the determinant by
using multi-linearity and anti-symmetry of determinants and eventually achieve
the result of Theorem \ref{thm:coefficients}-1.

\begin{lemma}
\label{lem:equiv}We have
\[
S_{\delta}=a_{0d_{0}}^{\delta_{0}}\cdot\det\left[
\begin{array}
[c]{rrr}%
\alpha_{1}^{0}F_{1}\left(  \alpha_{1}\right)   & \cdots & \alpha_{d_{0}}%
^{0}F_{1}\left(  \alpha_{d_{0}}\right)  \\
\vdots~~~~~~ &  & \vdots~~~~~~\\
\alpha_{1}^{\delta_{1}-1}F_{1}\left(  \alpha_{1}\right)   & \cdots &
\alpha_{d_{0}}^{\delta_{1}-1}F_{1}\left(  \alpha_{d_{0}}\right)  \\\hline
\vdots~~~~~~ &  & \vdots~~~~~~\\[-5pt]
\vdots~~~~~~ &  & \vdots~~~~~~\\\hline
\alpha_{1}^{0}F_{t}\left(  \alpha_{1}\right)   & \cdots & \alpha_{d_{0}}%
^{0}F_{t}\left(  \alpha_{d_{0}}\right)  \\
\vdots~~~~~~ &  & \vdots~~~~~~\\
\alpha_{1}^{\delta_{t}-1}F_{t}\left(  \alpha_{1}\right)   & \cdots &
\alpha_{d_{0}}^{\delta_{t}-1}F_{t}\left(  \alpha_{d_{0}}\right)  \\\hline
\alpha_{1}^{0}\left(  x-\alpha_{1}\right)   & \cdots & \alpha_{d_{0}}%
^{0}\left(  x-\alpha_{d_{0}}\right)  \\
\vdots~~~~~~ & & \vdots~~~~~~\\
\alpha_{1}^{\varepsilon-2}\left(  x-\alpha_{1}\right)   & \cdots &
\alpha_{d_{0}}^{\varepsilon-2}\left(  x-\alpha_{d_{0}}\right)
\end{array}
\right]  /V(\alpha_{1},\ldots,\alpha_{d_{0}})
\]

\end{lemma}

\begin{proof}
The lemma follows from the derivation below:

\begin{enumerate}
\item Let $N\ $be the numerator of the fractional part of $S_{\delta}$ (see Definition
\ref{def:roots}).

\item Note $N=\det\left[
\begin{array}
[c]{cc}%
U&\\
L_1&L_2
\end{array}
\right]  \ \ $where
\[
U=\left[
\begin{array}{rrr}
\alpha_{1}^{0}F_{1}\left(  \alpha_{1}\right)   & \cdots & \alpha_{d_{0}}%
^{0}F_{1}\left(  \alpha_{d_{0}}\right)  \\
\vdots~~~~~~~ &  & \vdots~~~~~~~\\
\alpha_{1}^{\delta_{1}-1}F_{1}\left(  \alpha_{1}\right)   & \cdots &
\alpha_{d_{0}}^{\delta_{1}-1}F_{1}\left(  \alpha_{d_{0}}\right)  \\\hline
\vdots~~~~~~~ &  & \vdots~~~~~~~\\[-5pt]
\vdots~~~~~~~ &  & \vdots~~~~~~~\\\hline
\alpha_{1}^{0}F_{t}\left(  \alpha_{1}\right)   & \cdots & \alpha_{d_{0}}%
^{0}F_{t}\left(  \alpha_{d_{0}}\right)  \\
\vdots~~~~~~~ &  & \vdots~~~~~~~\\
\alpha_{1}^{\delta_{t}-1}F_{t}\left(  \alpha_{1}\right)   & \cdots &
\alpha_{d_{0}}^{\delta_{t}-1}F_{t}\left(  \alpha_{d_{0}}\right)
\end{array}\right]
\quad
L_1=\left[
\begin{array}
[c]{ccc}
\alpha_{1}^{0} & \cdots & \alpha_{d_{0}}^{0} \\
\vdots & \vdots & \vdots \\
\alpha_{1}^{\varepsilon-1} & \cdots & \alpha_{d_{0}}^{\varepsilon-1} %
\end{array}
\right],
\qquad
L_2=\left[
\begin{array}
[c]{c}
 x^{0}\\
\vdots\\
x^{\varepsilon-1}%
\end{array}
\right]  \]

\item We will now cancel $x^{\varepsilon-1}$  by carrying out the row
operation on the last two rows of $L=\left[L_1 ~~L_2\right]$ as follows.%
\begin{align*}
N  & =\det\left[
\begin{array}
[c]{ccc|c}
& U &  & \\\hline
\alpha_{1}^{0} & \cdots & \alpha_{d_{0}}^{0} & x^{0}\\
\vdots & \vdots & \vdots & \vdots\\
\alpha_{1}^{\varepsilon-2} & \cdots & \alpha_{d_{0}}^{\varepsilon-2} &
x^{\varepsilon-2}\\
-\alpha_{1}^{\varepsilon-2}x+\alpha_{1}^{\varepsilon-1} & \cdots & -\alpha
_{d_{0}}^{\varepsilon-2}x+\alpha_{d_{0}}^{\varepsilon-1} & -x^{\varepsilon
-2}x-x^{\varepsilon-1}%
\end{array}
\right]  \\
& =\det\left[
\begin{array}
[c]{ccc|c}
& U &  & \\\hline
\alpha_{1}^{0} & \cdots & \alpha_{d_{0}}^{0} & x^{0}\\
\vdots & \vdots & \vdots & \vdots\\
\alpha_{1}^{\varepsilon-2} & \cdots & \alpha_{d_{0}}^{\varepsilon-2} &
x^{\varepsilon-2}\\
-\alpha_{1}^{\varepsilon-2}\left(  x-\alpha_{1}\right)   & \cdots &
-\alpha_{d_{0}}^{\varepsilon-2}\left(  x-\alpha_{d_{0}}\right)   & 0
\end{array}
\right]
\end{align*}

\item
By repeating the above operation for the other rows of $L$ except the first row, we get
\[N=\det\left[
\begin{array}
[c]{ccc|c}
& U &  & \\\hline
\alpha_{1}^{0} & \cdots & \alpha_{d_{0}}^{0} & x^{0}\\
-\alpha_{1}^{0}\left(  x-\alpha_{1}\right) & \cdots & -\alpha_{d_0}^{0}\left(  x-\alpha_{d_{0}}\right) &
\\
\vdots & \vdots & \vdots & \\
-\alpha_{1}^{\varepsilon-2}\left(  x-\alpha_{1}\right)   & \cdots &
-\alpha_{d_0}^{\varepsilon-2}\left(  x-\alpha_{d_{0}}\right)   &
\end{array}
\right]\]

\item
Note that there is only one non-zero entry $x^0$ in the last column of $N$ whose row and column indices are $1+\sum_{i=1}^t\delta_i$ and $d_0+1$ respectively. Thus we expand $N$ by the last column and get
\[N=(-1)^{(1+\sum_{i=1}^t\delta_i)+(d_0+1)}\cdot x^0\cdot \det\left[
\begin{array}
[c]{ccc}
& U &  \\\hline
-\alpha_{1}^{0}\left(  x-\alpha_{1}\right) & \cdots & -\alpha_{d_0}^{0}\left(  x-\alpha_{d_{0}}\right)
\\
\vdots & \vdots & \vdots \\
-\alpha_{1}^{\varepsilon-2}\left(  x-\alpha_{1}\right)   & \cdots &
-\alpha_{d_0}^{\varepsilon-2}\left(  x-\alpha_{d_{0}}\right)
\end{array}
\right]
\]
\item Next we extract the factor $-1$ from each of the last $\varepsilon-1$ row in the resulting determinant and achieve
\[N=(-1)^{\left(1+\sum_{i=1}^t\delta_i\right)+(d_0+1)}\cdot(-1)^{\varepsilon-1}\cdot x^0\cdot \det\left[
\begin{array}
[c]{ccc}
& U &  \\\hline
\alpha_{1}^{0}\left(  x-\alpha_{1}\right) & \cdots & \alpha_{d_0}^{0}\left(  x-\alpha_{d_{0}}\right)
\\
\vdots & \vdots & \vdots \\
\alpha_{1}^{\varepsilon-2}\left(  x-\alpha_{1}\right)   & \cdots &
\alpha_{d_0}^{\varepsilon-2}\left(  x-\alpha_{d_{0}}\right)
\end{array}
\right]
\]

\item Noting that
\[\left(1+\sum_{i=1}^t\delta_i\right)+(d_0+1)+(\varepsilon-1)=\left(\sum_{i=1}^t\delta_i+\epsilon\right)+(d_0+1)=2(d_0+1)\equiv0\mod 2\]
and combining $x^0=1$, we have
\[
N=\det\left[
\begin{array}
[c]{ccc}
& U &  \\\hline
\alpha_{1}^{0}\left(  x-\alpha_{1}\right) & \cdots & \alpha_{d_0}^{0}\left(  x-\alpha_{d_{0}}\right)
\\
\vdots & \vdots & \vdots \\
\alpha_{1}^{\varepsilon-2}\left(  x-\alpha_{1}\right)   & \cdots &
\alpha_{d_0}^{\varepsilon-2}\left(  x-\alpha_{d_{0}}\right)
\end{array}
\right]\]

\item
Thus $S_{\delta}=a_{0d_{0}}^{\delta_{0}}\cdot N/V(\alpha_1,\ldots,\alpha_{d_0})$.
\end{enumerate}
\end{proof}

\noindent Now combining Lemma \ref{lem:equiv} and proof techniques borrowed from
\cite{2007_DAndrea_Hong_Krick_Szanto} and \cite{2021_Hong_Yang}, we may prove
Theorem \ref{thm:coefficients}-1.

\begin{proof}[Proof of Theorem \ref{thm:coefficients}-1 (Sylvester-type)]
\

\begin{enumerate}
\item Recall
\[
M_{\delta}^{\text{Sylvester}}=\left[
\begin{array}
[c]{cccccccccc}%
\Psi_{q}\left( x^{0}F_{0}\right)  & \cdots & \Psi_{q}\left(
x^{\delta_{0}-1}F_{0}\right)  & \cdots & \cdots & \cdots & \Psi_{q}\left(
x^{0}F_{t}\right)  & \cdots & \Psi_{q}\left(  x^{\delta_{t}-1}%
F_{t}\right)  & X_{\delta,h}%
\end{array}
\right]  ^{T}%
\]

\item Let $d_{i}=\deg F_{i}$. Note that $\delta_{0}+d_{0}\ge\max\limits_{0\le
i\le t}(\delta_{i}+d_{i})$. Let $\tilde{d}_{i}=\delta
_{0}+d_{0}-\delta_{i}$. Then $\tilde{d}_{0}=d_{0}$.

\item Assume $F_{i}=\sum\limits_{0\leq j\leq\tilde{d}_{i}}a_{ij}x^{j}$. Then
we may write $\det M_{\delta}^{\text{Sylvester}}$ in the following explicit
form.
\begin{equation}
\det M_{\delta}^{\text{Sylvester}}=\det\left[
\begin{array}
[c]{cccccc}%
a_{00} & \cdots & \cdots & a_{0{d}_{0}} &  & \\
& \ddots & \ddots & \ddots & \ddots & \\
&  & a_{00} & \cdots & \cdots & a_{0{d}_{0}}\\
\ddots & \ddots & \ddots & \ddots & \ddots & \vdots\\
a_{t0} & \cdots & \cdots & a_{t\tilde{d}_{t}} &  & \\
& \ddots & \ddots & \ddots & \ddots & \\
&  & a_{t0} & \cdots & \cdots & a_{t\tilde{d}_{t}}\\
x & -1 &  &  &  & \\
& \ddots & \ddots &  &  & \\
&  & x & -1 &  &
\end{array}
\right] \hspace{-0.25in}
\begin{array}
[c]{l}%
\left.
\begin{array}
[c]{c}%
\\[10pt]%
\\
\end{array}
\right\} \delta_{0}\text{~rows}\\%
\begin{array}
[c]{c}%
\\[15pt]%
\end{array}
\\
\left.
\begin{array}
[c]{l}%
\\[10pt]%
\\
\end{array}
\right\} \delta_{t}\text{~rows}\\[20pt]%
\left.
\begin{array}
[c]{l}%
\\[10pt]%
\\
\end{array}
\right\} \epsilon-1\text{~rows}%
\end{array}
\label{eq:M_Sylvester}%
\end{equation}

\item Note that $F_{0}=a_{0d_{0}}\prod\limits_{i=1}^{d_{0}}\left(
x-\alpha_{i}\right)  $. Here we view $\alpha_{1},\ldots,\alpha_{d_{0}}$ as
distinct indeterminates. Thus
\[
V=V(\alpha_{1},\ldots,\alpha_{d_{0}})\neq0
\]
\item

In what follows, we show that the product of $\det M_{\delta}^{\text{Sylvester}}$ and $V$ is proportional to the numerator of the equivalent expression of $S_{\delta}$ in Lemma \ref{lem:equiv}.

\begin{enumerate}
\item We enlarge the size of $V$ so that it can match the size of $M_{\delta}^{\text{Sylvester}}$ while the determinant keeps unchanged.
\[\det M_{\delta}^{\text{Sylvester}}\cdot V\nonumber\\
=   \det\left[
\begin{array}
[c]{cccccc}%
a_{00} & \cdots & \cdots & a_{0{d_{0}}} &  & \\
& \ddots & \ddots & \ddots & \ddots & \\
&  & a_{00} & \cdots & \cdots & a_{0{d_{0}}}\\
\ddots & \ddots & \ddots & \ddots & \ddots & \vdots\\
a_{t0} & \cdots & \cdots & a_{t\tilde{d}_{t}} &  & \\
& \ddots & \ddots & \ddots & \ddots & \\
&  & a_{t0} & \cdots & \cdots & a_{t\tilde{d}_{t}}\\
x & -1 &  &  &  & \\
& \ddots & \ddots &  &  & \\
&  & x & -1 &  &
\end{array}
\right]  \cdot\det\left[
\begin{array}
[c]{cccccc}%
\alpha_{1}^{0} & \cdots & \alpha_{d_{0}}^{0} &  &  & \\
\vdots &  & \vdots &  &  & \\
\alpha_{1}^{d_{0}-1} &  & \alpha_{d_{0}}^{d_{0}-1} &  &  & \\
\alpha_{1}^{d_{0}} & \cdots & \alpha_{d_{0}}^{d_{0}} & 1 &  & \\
\vdots &  & \vdots &  & \ddots & \\
\alpha_{1}^{d_{0}+\delta_{0}-1} & \cdots & \alpha_{d_{0}}^{d_{0}+\delta_{0}-1}
&  &  & 1
\end{array}
\right] \]

\item Since $\det A\cdot\det B=\det AB$ for square matrices $A$ and $B$, we have
\[\det M_{\delta}^{\text{Sylvester}}\cdot V\nonumber\\
=   \det\left(\left[
\begin{array}
[c]{cccccc}%
a_{00} & \cdots & \cdots & a_{0{d_{0}}} &  & \\
& \ddots & \ddots & \ddots & \ddots & \\
&  & a_{00} & \cdots & \cdots & a_{0{d_{0}}}\\
\ddots & \ddots & \ddots & \ddots & \ddots & \vdots\\
a_{t0} & \cdots & \cdots & a_{t\tilde{d}_{t}} &  & \\
& \ddots & \ddots & \ddots & \ddots & \\
&  & a_{t0} & \cdots & \cdots & a_{t\tilde{d}_{t}}\\
x & -1 &  &  &  & \\
& \ddots & \ddots &  &  & \\
&  & x & -1 &  &
\end{array}
\right]  \cdot\left[
\begin{array}
[c]{cccccc}%
\alpha_{1}^{0} & \cdots & \alpha_{d_{0}}^{0} &  &  & \\
\vdots &  & \vdots &  &  & \\
\alpha_{1}^{d_{0}-1} &  & \alpha_{d_{0}}^{d_{0}-1} &  &  & \\
\alpha_{1}^{d_{0}} & \cdots & \alpha_{d_{0}}^{d_{0}} & 1 &  & \\
\vdots &  & \vdots &  & \ddots & \\
\alpha_{1}^{d_{0}+\delta_{0}-1} & \cdots & \alpha_{d_{0}}^{d_{0}+\delta_{0}-1}
&  &  & 1
\end{array}
\right]\right) \]

\item Carrying out matrix multiplication, we get
\[
\det M_{\delta}^{\text{Sylvester}}\cdot V=    \det\left[
\begin{array}
[c]{cccccc}%
\alpha_{1}^{0}F_{0}\left(  \alpha_{1}\right)  & \cdots & \alpha_{d_{0}}%
^{0}F_{0}\left(  \alpha_{d_{0}}\right)  & a_{0d_{0}} &  & \\
\vdots & \ddots & \vdots & \vdots & \ddots & \\
\alpha_{1}^{\delta_{0}-1}F_{0}\left(  \alpha_{1}\right)  & \cdots &
\alpha_{d_{0}}^{\delta_{0}-1}F_{0}\left(  \alpha_{d_{0}}\right)  & \cdot &
\cdots & a_{0d_{0}}\\
\alpha_{1}^{0}F_{1}\left(  \alpha_{1}\right)  & \cdots & \alpha_{d_{0}}%
^{0}F_{1}\left(  \alpha_{d_{0}}\right)  &  &  & \\
\vdots & \ddots & \vdots & \vdots & \ddots & \\
\alpha_{1}^{\delta_{1}-1}F_{1}\left(  \alpha_{1}\right)  & \cdots &
\alpha_{d_{0}}^{\delta_{1}-1}F_{1}\left(  \alpha_{d_{0}}\right)  & \cdot &
\cdots & a_{1\tilde{d}_{1}}\\
\vdots & \ddots & \vdots & \vdots & \ddots & \vdots\\
\alpha_{1}^{0}F_{t}\left(  \alpha_{1}\right)  & \cdots & \alpha_{d_{0}}%
^{0}F_{t}\left(  \alpha_{d_{0}}\right)  &  &  & \\
\vdots & \ddots & \vdots & \vdots & \ddots & \\
\alpha_{1}^{\delta_{t}-1}F_{t}\left(  \alpha_{1}\right)  & \cdots &
\alpha_{d_{0}}^{\delta_{t}-1}F_{t}\left(  \alpha_{d_{0}}\right)  & \cdot &
\cdots & a_{t\tilde{d}_{t}}\\
\alpha_{1}^{0}\left(  x-\alpha_{1}\right)  & \cdots & \alpha_{d_{0}}%
^{0}\left(  x-\alpha_{d_{0}}\right)  &  &  & \\
\vdots & \ddots & \vdots &  &  & \\
\alpha_{1}^{\varepsilon-2}\left(  x-\alpha_{1}\right)  & \cdots &
\alpha_{d_{0}}^{\varepsilon-2}\left(  x-\alpha_{d_{0}}\right)  &  &  &
\end{array}
\right]\]

\item Since $\alpha_1,\ldots,\alpha_{d_0}$ are roots of $F_0$, the above determinant is simplified:
\[
\det M_{\delta}^{\text{Sylvester}}\cdot V=    \det\left[
\begin{array}
[c]{cccccc}
&  &  & a_{0d_{0}} &  & \\
&  &  & \vdots & \ddots & \\
&  &  & \cdot & \cdots & a_{0d_{0}}\\
\alpha_{1}^{0}F_{1}\left(  \alpha_{1}\right)  & \cdots & \alpha_{d_{0}}%
^{0}F_{1}\left(  \alpha_{d_{0}}\right)  &  &  & \\
\vdots & \ddots & \vdots & \vdots & \ddots & \\
\alpha_{1}^{\delta_{1}-1}F_{1}\left(  \alpha_{1}\right)  & \cdots &
\alpha_{d_{0}}^{\delta_{1}-1}F_{1}\left(  \alpha_{d_{0}}\right)  & \cdot &
\cdots & a_{1\tilde{d}_{1}}\\
\vdots & \ddots & \vdots & \vdots & \ddots & \vdots\\
\alpha_{1}^{0}F_{t}\left(  \alpha_{1}\right)  & \cdots & \alpha_{d_{0}}%
^{0}F_{t}\left(  \alpha_{d_{0}}\right)  &  &  & \\
\vdots & \ddots & \vdots & \vdots & \ddots & \\
\alpha_{1}^{\delta_{t}-1}F_{t}\left(  \alpha_{1}\right)  & \cdots &
\alpha_{d_{0}}^{\delta_{t}-1}F_{t}\left(  \alpha_{d_{0}}\right)  & \cdot &
\cdots & a_{t\tilde{d}_{t}}\\
\alpha_{1}^{0}\left(  x-\alpha_{1}\right)  & \cdots & \alpha_{d_{0}}%
^{0}\left(  x-\alpha_{d_{0}}\right)  &  &  & \\
\vdots & \ddots & \vdots &  &  & \\
\alpha_{1}^{\varepsilon-2}\left(  x-\alpha_{1}\right)  & \cdots &
\alpha_{d_{0}}^{\varepsilon-2}\left(  x-\alpha_{d_{0}}\right)  &  &  &
\end{array}
\right]
\]

\item Note that in the first $\delta_0$ rows of the righthand side determinant, there is only one nonzero minor of order $\delta_0$. The entries of this minor lie in the intersections of the first $\delta_0$ rows and the last $\delta_0$ columns (i.e,, $(d_0+1)$th, $(d_0+2)$th, $\ldots$, $(d_0+\delta_0)$th columns) and obviously the minor has a diagonal structure. Thus after taking Laplace expansion for the determinant, we get
\begin{align*}
&\det M_{\delta}^{\text{Sylvester}}\cdot V\\
=  &  (-1)^{\sum_{i=1}^{\delta_{0}}{i}+\sum_{i=d_{0}+1}^{d_{0}+\delta_{0}}{i}%
}\cdot \det
\left[
\begin{array}{ccc}
a_{0d_0} &  &  \\
\vdots & \ddots &  \\
\cdot & \cdots & a_{0d_0} \\
\end{array}
\right]
\cdot\det\left[
\begin{array}
[c]{ccc}%
\alpha_{1}^{0}F_{1}\left(  \alpha_{1}\right)  & \cdots & \alpha_{d_{0}}%
^{0}F_{1}\left(  \alpha_{d_{0}}\right) \\
\vdots & \ddots & \vdots\\
\alpha_{1}^{\delta_{1}-1}F_{1}\left(  \alpha_{1}\right)  & \cdots &
\alpha_{d_{0}}^{\delta_{1}-1}F_{1}\left(  \alpha_{d_{0}}\right) \\
\vdots & \ddots & \vdots\\
\alpha_{1}^{0}F_{t}\left(  \alpha_{1}\right)  & \cdots & \alpha_{d_{0}}%
^{0}F_{t}\left(  \alpha_{d_{0}}\right) \\
\vdots & \ddots & \vdots\\
\alpha_{1}^{\delta_{t}-1}F_{t}\left(  \alpha_{1}\right)  & \cdots &
\alpha_{d_{0}}^{\delta_{t}-1}F_{t}\left(  \alpha_{d_{0}}\right) \\
\alpha_{1}^{0}\left(  x-\alpha_{1}\right)  & \cdots & \alpha_{d_{0}}%
^{0}\left(  x-\alpha_{d_{0}}\right) \\
\vdots & \ddots & \vdots\\
\alpha_{1}^{\varepsilon-2}\left(  x-\alpha_{1}\right)  & \cdots &
\alpha_{d_{0}}^{\varepsilon-2}\left(  x-\alpha_{d_{0}}\right)
\end{array}
\right]\\
=  &  (-1)^{\sum_{i=1}^{\delta_{0}}{i}+\sum_{i=d_{0}+1}^{d_{0}+\delta_{0}}{i}%
}\cdot a_{0,d_{0}}^{\delta_{0}}\cdot\det\left[
\begin{array}
[c]{ccc}%
\alpha_{1}^{0}F_{1}\left(  \alpha_{1}\right)  & \cdots & \alpha_{d_{0}}%
^{0}F_{1}\left(  \alpha_{d_{0}}\right) \\
\vdots & \ddots & \vdots\\
\alpha_{1}^{\delta_{1}-1}F_{1}\left(  \alpha_{1}\right)  & \cdots &
\alpha_{d_{0}}^{\delta_{1}-1}F_{1}\left(  \alpha_{d_{0}}\right) \\
\vdots & \ddots & \vdots\\
\alpha_{1}^{0}F_{t}\left(  \alpha_{1}\right)  & \cdots & \alpha_{d_{0}}%
^{0}F_{t}\left(  \alpha_{d_{0}}\right) \\
\vdots & \ddots & \vdots\\
\alpha_{1}^{\delta_{t}-1}F_{t}\left(  \alpha_{1}\right)  & \cdots &
\alpha_{d_{0}}^{\delta_{t}-1}F_{t}\left(  \alpha_{d_{0}}\right) \\
\alpha_{1}^{0}\left(  x-\alpha_{1}\right)  & \cdots & \alpha_{d_{0}}%
^{0}\left(  x-\alpha_{d_{0}}\right) \\
\vdots & \ddots & \vdots\\
\alpha_{1}^{\varepsilon-2}\left(  x-\alpha_{1}\right)  & \cdots &
\alpha_{d_{0}}^{\varepsilon-2}\left(  x-\alpha_{d_{0}}\right)
\end{array}
\right]\\
=& (-1)^{\sum_{i=1}^{\delta_{0}}{i}+\sum_{i=d_{0}+1}^{d_{0}+\delta_{0}}{i}%
}\cdot S_{\delta}\cdot V
\end{align*}

\item Since
$$\sum_{i=1}^{\delta_{0}}{i}+\sum_{i=d_{0}+1}^{d_{0}+\delta_{0}}{i}
=\sum_{i=1}^{\delta_{0}}{i}+\sum_{i=1}^{\delta_{0}}{(d_0+i)}
=d_0\delta_{0}+2\sum_{i=1}^{\delta_{0}}{i}\equiv d_0\delta_0 \mod 2,$$
we have
\begin{align*}
\det M_{\delta}^{\text{Sylvester}}\cdot V&=  (-1)^{d_{0}\delta_{0}}\cdot S_{\delta}\cdot V \nonumber
\end{align*}
\end{enumerate}

\item Finally we have
\[
S_{\delta}={(-1)^{d_{0}\delta_{0}}}\cdot\det
M_{\delta}^{\text{Sylvester}}%
\]

\end{enumerate}
\end{proof}

\subsection{ Proof of Theorem \ref{thm:coefficients}-2 (Barnett-type)}

Here is a high level view of the proof. The main technical challenge is how to
connect the extended Barnett matrix with the roots of the first polynomial.
For this purpose, we  multiply the extended Barnett
matrix with the Vandermonde matrix of $\alpha_{1},\ldots,\alpha_{d_{0}}$.
Observing that the companion matrix of a univariate polynomial $A(x)$ of
degree $n$ represents the endomorphism of of $\mathbb{C}_{< n}[x]$ defined by
the multiplication of $x$ (see Lemma \ref{lem:deltaA_linear_map}) and
generalizing it to the multiplication of an arbitrary univariate polynomial
(see Lemma \ref{lem:deltaA_polynomial_version}), we may convert the product
into an expression in terms of roots, which is exactly the numerator of
 $S_{\delta}$ in Lemma \ref{lem:equiv}.
We start the proof with recalling the following well-known lemma.

\begin{lemma}
\label{lem:deltaA_linear_map} Let $A=\sum_{i=0}^{n}a_{i}x^{i}$ and $C_{A}$ be
as in Notation \ref{notation:matrices}. Then
\[
\bar{x}x\equiv_{A}\bar{x}C_{A} \text{~~~where~~~} \bar{x}=\left[
\begin{array}
[c]{lll}%
x^{0} & \cdots & x^{n-1}%
\end{array}
\right]
\]

\end{lemma}


When $x$ is generalized to a univariate polynomial $H$, Lemma
\ref{lem:deltaA_linear_map} can be generalized to the following Lemma
\ref{lem:deltaA_polynomial_version} which is the essential ingredient for
building up the connection between the extended Barnett matrix and the roots
of $A$.

\begin{lemma}
\label{lem:deltaA_polynomial_version} $\bar{x}H\left(
C_{A}\right) \equiv_{A} \bar{x}H$.
\end{lemma}

\begin{proof}
It is easy to verify that%
\[
\bar{x}H\left(  C_{A}\right)  =\bar{x}\sum_{j\geq0}h_{j}C_{A}^{j}=\sum
_{j\geq0}h_{j}\bar{x}C_{A}^{j}\equiv_{A}\sum_{j\geq0}h_{j}\bar{x}x^{j}=\bar
{x}\left(  \sum_{j\geq0}h_{j}x^{j}\right)  =\bar{x}H
\]

\end{proof}

\begin{remark}\label{rem:companion_matrix}
When $x$ is specialized to a root of $A$, say $\alpha_{i}$, we get
$\bar{\alpha}_{i}H(C_{A})=\bar{\alpha}_{i}H(\alpha_{i})$ where $\bar{\alpha
}_{i}=(\alpha_{i}^{0},\ldots,\alpha_{i}^{n-1}).$
\end{remark}

\noindent Now we are ready to prove Theorem \ref{thm:coefficients}-2.

\begin{proof}[Proof of Theorem \ref{thm:coefficients}-2 (Barnett-type)]\

\begin{enumerate}
\item Specialize $A$, $H$ and $x$
in Lemma \ref{lem:deltaA_polynomial_version} with $F_{0}$, $F_{k}$ and $\alpha_{i}$, respectively, where $\alpha_i$'s are the roots of $F_0$. Again let $\bar{\alpha}_{i}=(\alpha_{i}^{0},\ldots,\alpha_{i}^{d_{0}-1})$.
Then by Lemma \ref{lem:deltaA_polynomial_version} and Remark \ref{rem:companion_matrix}, we get
\[
\bar{\alpha}_{i}F_{k}(C_{F_{0}})=\bar{\alpha}_{i}F_{k}(\alpha_{i})
\]

\item Recall that $R_{kj}$ is the ${j}$-th column of $F_{k}(C_{F_{0}})$. Thus selecting the $j$-th column on both sides in the above equation, we have
\begin{equation}\label{eqs:translation}
\bar{\alpha}_{i}R_{kj}=\alpha_{i}^{j-1}F_{k}(\alpha_{i}).
\end{equation}

\item Assembling $\bar{\alpha}_{i}~(1\le i\le d_0)$ vertically and $R_{kj}~(1\le j\le \delta_k)$ horizontally, we
obtain two matrices, denoted by $\bar{\alpha}$ and $R_k$, i.e,
$$\bar{\alpha}=\left[
\begin{array}
[c]{c}%
\bar{\alpha}_{1}\\
\vdots\\
\bar{\alpha}_{d_{0}}%
\end{array}
\right],\qquad R_k=\left[
\begin{array}
[c]{ccc}%
R_{k1} & \ldots & R_{k\delta_{k}}%
\end{array}
\right]$$

\item Multiplying $\bar{\alpha}$ and $R_k$, we obtain
\[
\bar{\alpha}R_k=\left[
\begin{array}
[c]{c}%
\bar{\alpha}_{1}\\
\vdots\\
\bar{\alpha}_{d_{0}}%
\end{array}
\right]  \left[
\begin{array}
[c]{ccc}%
R_{k1} & \ldots & R_{k\delta_{k}}%
\end{array}
\right]
=\left[
\begin{array}
[c]{ccc}%
\bar{\alpha}_{1}R_{k1} & \cdots & \bar{\alpha}_{1}R_{k\delta_k}\\
\vdots & \ddots & \vdots\\
\bar{\alpha}_{d_{0}}R_{k1} & \cdots & \bar{\alpha}_{d_{0}}R_{k\delta_k}
\end{array}
\right]\]

\item  Substituting \eqref{eqs:translation} into the resulting matrix yields
\[\bar{\alpha}R_k=\left[
\begin{array}
[c]{ccc}%
\alpha_{1}^{0}F_{k}(\alpha_{1}) & \cdots & \alpha_{1}^{\delta_{k}-1}%
F_{k}(\alpha_{1})\\
\vdots & \ddots & \vdots\\
\alpha_{d_{0}}^{0}F_{k}(\alpha_{d_{0}}) & \cdots & \alpha_{d_{0}}^{\delta
_{k}-1}F_{k}(\alpha_{d_{0}})
\end{array}
\right]
\]

\item Next we simplify $\bar{\alpha}X_{\delta,d_0}$. Recall that
\[
X_{\delta,d_0}=%
\begin{array}
[c]{l}%
\begin{bmatrix}
x &  &  & \\
-1 & \ddots &  & \\
& \ddots & \ddots & \\
&  & \ddots & x\\
&  &  & -1\\
&  &  & \\
&  &  &
\end{bmatrix}
\hspace{-1em}
\left.
\begin{array}
[c]{c}%
\;\\
\;\\
\;\\
\;\\
\;\\
\;\\
\;\\
\
\end{array}
\right\}  d_0~\text{rows}\\
\hspace{-1em}\underbrace{\hspace{8em}}_{d_{0}-\left(  \delta_{1}+\cdots+\delta_{t}\right)
=\varepsilon-1
~\text{columns}~}%
\end{array}\]By calculation, we have
\begin{align}
\bar{\alpha}X_{\delta,d_0}
&=\left[
\begin{array}{ccc}%
\alpha_{1}^0&\cdots&{\alpha}_{1}^{d_0-1}\\
\vdots&&\vdots\\
\alpha_{d_{0}}^0&\cdots&\alpha_{d_{0}}^{d_0-1}
\end{array}
\right]
\begin{bmatrix}
x &  &  & \\
-1 & \ddots &  & \\
& \ddots & \ddots & \\
&  & \ddots & x\\
&  &  & -1\\
&  &  & \\
&  &  &
\end{bmatrix}\nonumber\\
&=
\left[
\begin{array}{ccc}
\alpha_{1}^{0}x-\alpha_{1}^{1} & \cdots & \alpha_{1}^{\varepsilon-2}x-\alpha_{1}^{\varepsilon-1} \\
\vdots &  & \vdots \\
\alpha_{d_0}^{0}x-\alpha_{d_0}^{1} & \cdots & \alpha_{d_0}^{\varepsilon-2}x-\alpha_{d_0}^{\varepsilon-1}
\end{array}
\right]\nonumber\\
&=\left[
\begin{array}
[c]{ccc}%
\alpha_{1}^{0}\left(  x-\alpha_{1}\right)  & \cdots & \alpha_{1}%
^{\varepsilon-2}\left(  x-\alpha_{1}\right) \\
\vdots & \ddots & \vdots\\
\alpha_{d_{0}}^{0}\left(  x-\alpha_{d_{0}}\right)  & \cdots & \alpha_{d_{0}%
}^{\varepsilon-2}\left(  x-\alpha_{d_{0}}\right)
\end{array}
\right]\label{eqs:alpha*X}
\end{align}

\item Now assembling $R_k~(1\le k\le t)$ and $X_{\delta,d_{0}}$ horizontally and pre-multiplying the resulting matrix by $
\bar{\alpha}$, we get
\begin{align*}
\text{LHS}=&\bar{\alpha} \left[
\begin{array}
[c]{cccc}%
R_{1} & \cdots &  R_{t}& X_{\delta,d_0}%
\end{array}
\right] \\
=&\left[
\begin{array}
[c]{cccc}%
\bar{\alpha}R_{1} & \cdots &  \bar{\alpha}R_{t}& \bar{\alpha}X_{\delta,d_0}%
\end{array}
\right]\\
=  &  {\footnotesize {\,\,\left[
\begin{array}
[c]{ccc|c|ccc|}%
\alpha_{1}^{0}F_{1}(\alpha_{1}) & \cdots & \alpha_{1}^{\delta_{1}-1}%
F_{1}(\alpha_{1}) & \cdots\cdots & \alpha_{1}^{0}F_{t}(\alpha_{1}) & \cdots &
\alpha_{1}^{\delta_{t}-1}F_{t}(\alpha_{1}) \\
\vdots & & \vdots & & \vdots &  & \vdots \\
\alpha_{d_{0}}^{0}F_{1}(\alpha_{d_{0}}) & \cdots & \alpha_{d_{0}}^{\delta
_{1}-1}F_{1}(\alpha_{d_{0}}) & \cdots\cdots & \alpha_{d_{0}}^{0}F_{t}(\alpha_{d_{0}%
}) & \cdots & \alpha_{d_{0}}^{\delta_{t}-1}F_{t}(\alpha_{d_{0}})
\end{array}
\right.  }}\\
  &  \hspace{23em}{\footnotesize {\,\,\left.
\begin{array}
[c]{cccccccccc}%
\alpha_{1}^{0}\left(
x-\alpha_{1}\right)  & \cdots & \alpha_{1}^{\varepsilon-2}\left(  x-\alpha
_{1}\right) \\
\vdots &
& \vdots\\
\alpha_{{d_{0}}}^{0}\left(  x-\alpha_{{d_{0}}}\right)  & \cdots &
\alpha_{{d_{0}}}^{\varepsilon-2}\left(  x-\alpha_{{d_{0}}}\right)
\end{array}
\right]  }}\\
=&\text{RHS}
\end{align*}

\item Taking transpose and computing the determinants for the lefthand and righthand side expressions yield
\begin{align*}
\det \text{LHS}&=\det( \left[
\begin{array}
[c]{cccc}%
R_{1} & \cdots &  R_{t}& X_{\delta,d_0}%
\end{array}
\right]^T\cdot \bar{\alpha}^T)=\det(M_{\delta}^{\text{Barnett}})\cdot V\\
\det \text{RHS}&=\det\left[
\begin{array}
[c]{ccc}%
F_{1}(\alpha_{1})\alpha_{1}^{0} & \cdots & F_{1}(\alpha_{d_{0}})\alpha_{d_{0}%
}^{0}\\
\vdots & \ddots & \vdots\\
F_{1}(\alpha_{1})\alpha_{1}^{\delta_{1}-1} & \cdots & F_{1}(\alpha_{d_{0}%
})\alpha_{d_{0}}^{\delta_{1}-1}\\\hline
\vdots & \ddots & \vdots\\
\vdots & \ddots & \vdots\\\hline
F_{t}(\alpha_{1})\alpha_{1}^{0} & \cdots & F_{t}(\alpha_{d_{0}})\alpha_{d_{0}%
}^{0}\\
\vdots & \ddots & \vdots\\
F_{t}(\alpha_{1})\alpha_{1}^{\delta_{t}-1} & \cdots & F_{t}(\alpha_{d_{0}%
})\alpha_{d_{0}}^{\delta_{t}-1}\\\hline
\alpha_{1}^{0}\left(  x-\alpha_{1}\right)  & \cdots & \alpha_{{d_{0}}}%
^{0}\left(  x-\alpha_{{d_{0}}}\right) \\
\vdots & \ddots & \vdots\\
\alpha_{1}^{\varepsilon-2}\left(  x-\alpha_{1}\right)  & \cdots &
\alpha_{{d_{0}}}^{\varepsilon-2}\left(  x-\alpha_{{d_{0}}}\right)
\end{array}
\right]
=    a_{0d_{0}}^{-\delta_{0}}\cdot V{S}_{\delta}
\end{align*}

\item Finally we have  $S_{\delta}=a_{0d_{0}}^{\delta_{0}}\cdot\det M_{\delta}^{\text{Barnett}}$.
\end{enumerate}
\end{proof}

\subsection{Proof of Theorem \ref{thm:coefficients}-3 (B\'{e}zout type)}

Here is a high level view of the proof. First, by multiplying the extended
B\'{e}zout matrix with the Vandmonde matrix of $\alpha_{1},\ldots
,\alpha_{d_{0}}$ and exploring the connection between the B\'{e}zout matrix
and roots of the involved polynomials, we convert the product into an
expression in terms of roots. Then we repeatedly rewrite the determinant of
the obtained product matrix by using multi-linearity and anti-symmetry of
determinants and achieve the result.
We start with introducing some notations for constructing B\'{e}zout matrices.

\begin{notation}
\ \label{notation2}

\begin{itemize}
\item $A=\sum_{i=0}^{n}a_{i}x^{i}=a_{n}\prod_{i=1}^{n}(x-\alpha_{i})$ and
$B\in\mathbb{C}[x]$ such that $\deg A\ge\deg B$

\item $\bar{\alpha}_{i} = (\alpha_{i}^{0},\ldots,\alpha_{i}^{n-1})$

\item $M$ is the B\'{e}zout matrix of $A$ and $B$

\item $M_{j}$ is the $j$-th column of $M$

\item $e_{j}$ is the $j$-th elementary symmetric polynomials in $\alpha
_{1},\ldots,\alpha_{n}$ (with $e_{0}:=1$)

\item $e_{j}^{(i)}$ is the $j$-th
elementary symmetric polynomials in $\alpha_{1},\ldots,\widehat{\alpha_{i}%
},\ldots,\alpha_{n}$ (with $e_{0}^{(i)}:=1$).
\end{itemize}
\end{notation}

\noindent Then we may establish the following lemmas, which will be needed for
the proof of Theorem \ref{thm:coefficients}-3.

\begin{lemma}
\label{lem:v_iM_j} $\bar{\alpha}_{i}M_{j}$ is the coefficient of $a_{n}B(\alpha_{i}%
)\prod_{k\neq i}{(x-\alpha_{k})}$ for the term $x^{n-j}$, i.e.,
$$\bar{\alpha}_{i}M_{j}=a_{n}B(\alpha_{i})\cdot
(-1)^{j-1}e_{j-1}^{(i)}$$
\end{lemma}

\begin{proof}
By the definition of B\'{e}zout matrix,
\[
\dfrac{\det\left[
\begin{array}
[c]{cc}%
A(x) & A(y)\\
B(x) & B(y)
\end{array}
\right]  }{x-y} =\det\left[
\begin{array}
[c]{cc}%
\dfrac{A(x)-A(y)}{x-y} & A(y)\\[10pt]%
\dfrac{B(x)-B(y)}{x-y} & B(y)
\end{array}
\right]  =[%
\begin{array}
[c]{ccc}%
y^{0} & \cdots & y^{n-1}%
\end{array}
]M \left[
\begin{array}
[c]{c}%
x^{n-1}\\
\vdots\\
x^{0}%
\end{array}
\right]
\]
The substitution of $y=\alpha_{i}$ into the left and right sides of the above expression yields

\begin{enumerate}
\item $\text{LHS}=\det\left[
\begin{array}
[c]{cc}%
\dfrac{A(x)-A(\alpha_{i})}{x-\alpha_{i}} & 0\\[10pt]%
\dfrac{B(x)-B(\alpha_{i})}{x-\alpha_{i}} & B(\alpha_{i})
\end{array}
\right]  = B(\alpha_{i})\cdot\dfrac{A(x)-A(\alpha_i)}{x-\alpha_{i}}
=B(\alpha_{i})\cdot\dfrac{A(x)}{x-\alpha_{i}}=a_{n}B(\alpha_{i})\prod\limits_{k\ne i}(x-\alpha_{k})$

\item $\text{RHS}=\bar{\alpha}_{i}[
\begin{array}
[c]{ccc}%
M_{1} & \cdots & M_{n}%
\end{array}
]\left[
\begin{array}
[c]{c}%
x^{n-1}\\
\vdots\\
x^{0}%
\end{array}
\right]  =[
\begin{array}
[c]{ccc}%
\bar{\alpha}_{i}M_{1} & \cdots & \bar{\alpha}_{i}M_{n}%
\end{array}
]\left[
\begin{array}
[c]{c}%
x^{n-1}\\
\vdots\\
x^{0}%
\end{array}
\right] \qquad=\sum_{j=0}^{n-1}\bar{\alpha}_{i}M_{n-j}x^{j}$
\end{enumerate}

\noindent Comparing the coefficients of $x^{n-j}$ in the two expressions, we
get
\[
\bar{\alpha}_{i}M_{j}=a_{n}B(\alpha_{i})\cdot(-1)^{j-1}e_{j-1}^{(i)}
\]
\end{proof}

\noindent Next we explore the connection between $e_{j}^{(i)}$ and $e_{k}~(1\le k\le n)$.
\begin{lemma}
We have%
\begin{equation}\label{eqs:ej_exclude_i}
e_{j}^{(i)}=\sum_{k=0}^{j}(-1)^{k}e_{j-k}\alpha_{i}^{k},\quad i=1,\ldots,n,~
j=0,\ldots,n-1
\end{equation}
\end{lemma}

\begin{proof}
The proof is given in an inductive way for $j$ based on the observation
$e_{j}^{(i)}=e_{j}-\alpha_{i}e_{j-1}^{(i)}$ when $j>0$.

\begin{itemize}
\item When $j=0$, it is easy to verify that $LHS=e_{0}^{(i)}=e_{0}=(-1)^{0}%
e_{0}\alpha_{i}^{0}=RHS$.

\item Assume that the equation holds for $j<1$. Then
\begin{align*}
e_{j}^{(i)}  &  =e_{j}-\alpha_{i}e_{j-1}^{(i)}\\
&  =e_{j}-\alpha_{i}\sum_{k=0}^{j-1}(-1)^{k}e_{j-1-k}\alpha_{i}^{k}\\
&  =e_{j}+\sum_{k=0}^{j-1}(-1)^{k+1}e_{j-(k+1)}\alpha_{i}^{k+1}\\
&  =e_{j}+\sum_{k=1}^{j}(-1)^{k}e_{j-k}\alpha_{i}%
^{k}\\
&  =(-1)^{0}e_{j-0}\alpha_{i}^{0}+\sum_{k=1}^{j}(-1)^{k}e_{j-k}\alpha_{i}%
^{k}\\
&  =\sum_{k=0}^{j}(-1)^{k}e_{j-k}\alpha_{i}^{k}%
\end{align*}

\end{itemize}
\end{proof}

\noindent Now we are ready to prove Theorem \ref{thm:coefficients}-3.

\begin{proof}[Proof of Theorem \ref{thm:coefficients}-3 (B\'{e}zout-type)]\
\

\begin{enumerate}
\item Recall that $\bar{\alpha}_{i}=(\alpha_{i}^{0},\ldots,\alpha_{i}^{d_{0}-1})$ and $R_{kj}$ is the $j$th column of the
B\'{e}zout matrix of $F_{0}$ and $F_{k}$. Again we assemble $\bar{\alpha}_{i}~(1\le i\le d_0)$ vertically and $R_{kj}~(1\le j\le \delta_k)$ horizontally and obtain two matrices as follows:
\[\bar{\alpha}=\left[
\begin{array}
[c]{c}%
\bar{\alpha}_{1}\\
\vdots\\
\bar{\alpha}_{d_{0}}%
\end{array}
\right],\quad R_k=  \left[
\begin{array}
[c]{ccc}%
R_{k1} & \ldots & R_{k\delta_{k}}%
\end{array}
\right]  \]
\item
Now multiplying $\bar{\alpha}$ and $R_k$ and simplifying the resulting expression by Lemma \ref{lem:v_iM_j}, we get
\begin{align*}
\bar{\alpha}R_{k}=&\left[
\begin{array}
[c]{c}%
\bar{\alpha}_{1}\\
\vdots\\
\bar{\alpha}_{d_{0}}%
\end{array}
\right]  \left[
\begin{array}
[c]{ccc}%
R_{k1} & \ldots & R_{k\delta_{k}}%
\end{array}
\right]\\
=&\left[
\begin{array}
[c]{ccc}%
\bar{\alpha}_{1}R_{k1} & \cdots & \bar{\alpha}_{1}R_{k\delta_k}\\
\vdots & \ddots & \vdots\\
\bar{\alpha}_{d_{0}}R_{k1} & \cdots & \bar{\alpha}_{d_{0}}R_{k\delta_k}
\end{array}
\right]\\
=&\left[
\begin{array}
[c]{ccc}%
a_{0d_{0}}F_{k}(\alpha_{1})e_{0}^{(1)} & \cdots & (-1)^{\delta_{k}-1}%
a_{0d_{0}}F_{k}(\alpha_{1})e_{\delta_{k}-1}^{(1)}\\
\vdots & \ddots & \vdots\\
a_{0d_{0}}F_{k}(\alpha_{d_{0}})e_{0}^{(d_{0})} & \cdots & (-1)^{\delta_{k}%
-1}a_{0d_{0}}F_{k}(\alpha_{d_{0}})e_{\delta_{k}-1}^{(d_{0})}%
\end{array}
\right]
\end{align*}

\item Recall \eqref{eqs:alpha*X},
i.e.,
\[\bar{\alpha}X_{\delta,d_0}
=\left[
\begin{array}
[c]{ccc}%
\alpha_{1}^{0}\left(  x-\alpha_{1}\right)  & \cdots & \alpha_{1}%
^{\varepsilon-2}\left(  x-\alpha_{1}\right) \\
\vdots & \ddots & \vdots\\
\alpha_{d_{0}}^{0}\left(  x-\alpha_{d_{0}}\right)  & \cdots & \alpha_{d_{0}%
}^{\varepsilon-2}\left(  x-\alpha_{d_{0}}\right)
\end{array}
\right]
\]

\item Now assembling $R_k~(1\le k\le t)$ and $X_{\delta,d_{0}}$ horizontally and pre-multiplying the resulting matrix by $
\bar{\alpha}$, we get
\begin{align*}
\text{LHS}=&\bar{\alpha} \left[
\begin{array}
[c]{cccc}%
R_{1} & \cdots &  R_{t}& X_{\delta,d_0}%
\end{array}
\right] \\
=&\left[
\begin{array}
[c]{cccc}%
\bar{\alpha}R_{1} & \cdots &  \bar{\alpha}R_{t}& \bar{\alpha}X_{\delta,d_0}%
\end{array}
\right]\\
=  &  {\footnotesize {\,\,\left[
\begin{array}
[c]{ccc|c}%
(-1)^{0}a_{0d_{0}}F_{1}(\alpha_{1})e_{0}^{(1)} & \cdots & (-1)^{\delta_{1}-1}%
a_{0d_{0}}F_{1}(\alpha_{1})e_{\delta_{1}-1}^{(1)} & \cdots\cdots  \\
\vdots & & \vdots & \\
(-1)^{0}a_{0d_{0}}F_{1}(\alpha_{d_0})e_{0}^{(d_0)} & \cdots & (-1)^{\delta_{1}-1}%
a_{0d_{0}}F_{1}(\alpha_{d_0})e_{\delta_{d_0}-1}^{(d_0)} & \cdots \cdots
\end{array}
\right.  }}\\
=  &  \hspace{1em}{\footnotesize {\,\,\left.
\begin{array}
[c]{c|ccc|ccc}%
\cdots\cdots&(-1)^{0}a_{0d_{0}}F_{t}(\alpha_{1})e_{0}^{(1)} & \cdots & (-1)^{\delta_{t}-1}%
a_{0d_{0}}F_{t}(\alpha_{1})e_{\delta_{t}-1}^{(1)} & \alpha_{1}^{0}\left(
x-\alpha_{1}\right)  & \cdots & \alpha_{1}^{\varepsilon-2}\left(  x-\alpha
_{1}\right) \\
&\vdots & & \vdots&\vdots & & \vdots\\
\cdots\cdots&(-1)^{0}a_{0d_{0}}F_{t}(\alpha_{d_0})e_{0}^{(d_0)} & \cdots & (-1)^{\delta_{t}-1}%
a_{0d_{0}}F_{t}(\alpha_{d_0})e_{\delta_{t}-1}^{(d_0)}&\alpha_{{d_{0}}}^{0}\left(  x-\alpha_{{d_{0}}}\right)  & \cdots &
\alpha_{{d_{0}}}^{\varepsilon-2}\left(  x-\alpha_{{d_{0}}}\right)
\end{array}
\right]  }}\\
=&\text{RHS}
\end{align*}

\item Taking transpose and computing the determinants for the left-hand and right-hand side expressions yield
\begin{align*}
\det \text{LHS}&=\det( \left[
\begin{array}
[c]{cccc}%
R_{1} & \cdots &  R_{t}& X_{\delta,d_0}%
\end{array}
\right]^T\cdot \bar{\alpha}^T)=\det M_{\delta}^{\text{B\'{e}zout}}\cdot V\\
\det \text{RHS}&=\det\left[
\begin{array}
[c]{ccc}%
(-1)^{0}a_{0d_{0}}F_{1}(\alpha_{1})e_{0}^{(1)} & \cdots & (-1)^{0}a_{0d_{0}}F_{1}%
(\alpha_{d_{0}})e_{0}^{(d_{0})}\\
\vdots & & \vdots\\
(-1)^{\delta_{1}-1}a_{0d_{0}}F_{1}(\alpha_{1})e_{\delta_{1}-1}^{(1)} & \cdots
& (-1)^{\delta_{1}-1}a_{0d_{0}}F_{1}(\alpha_{d_{0}})e_{\delta_{1}-1}^{(d_{0}%
)}\\\hline
\vdots & & \vdots\\
\vdots &  & \vdots\\\hline
(-1)^{0}a_{0d_{0}}F_{t}(\alpha_{1})e_{0}^{(1)} & \cdots & (-1)^{0}a_{0d_{0}}F_{t}%
(\alpha_{d_{0}})e_{0}^{(d_{0})}\\
\vdots &  & \vdots\\
(-1)^{\delta_{t}-1}a_{0d_{0}}F_{t}(\alpha_{1})e_{\delta_{t}-1}^{(1)} & \cdots
& (-1)^{\delta_{t}-1}a_{0d_{0}}F_{t}(\alpha_{d_{0}})e_{\delta_{t}-1}^{(d_{0}%
)}\\\hline
\alpha_{1}^{0}\left(  x-\alpha_{1}\right)  & \cdots & \alpha_{{d_{0}}}%
^{0}\left(  x-\alpha_{{d_{0}}}\right) \\
\vdots & & \vdots\\
\alpha_{1}^{\varepsilon-2}\left(  x-\alpha_{1}\right)  & \cdots &
\alpha_{{d_{0}}}^{\varepsilon-2}\left(  x-\alpha_{{d_{0}}}\right)
\end{array}
\right]
\end{align*}

\item In what follows, we rewrite $\det \text{RHS}$ by using the multi-linearity and anti-symmetry properties of determinant.
\begin{enumerate}
\item Extract the factor $(-1)^{k_i}a_{0d_{0}}$ from the $i$th row ($1\le i\le \delta_1+\cdots+\delta_t$) where $k_i$ is some natural number depending on $i$.
\[
\det \text{RHS}=    (-1)^{\sum_{i=1}^{t}{(0+\cdots+\delta_{i}-1)}}\cdot a_{0d_{0}}^{\sum_{i=1}%
^{t}\delta_{i}}\cdot\det\left[
\begin{array}
[c]{ccc}%
F_{1}(\alpha_{1})e_{0}^{(1)} & \cdots & F_{1}(\alpha_{d_{0}})e_{0}^{(d_{0})}\\
\vdots & & \vdots\\
F_{1}(\alpha_{1})e_{\delta_{1}-1}^{(1)} & \cdots & F_{1}(\alpha_{d_{0}%
})e_{\delta_{1}-1}^{(d_{0})}\\\hline
\vdots & & \vdots\\
\vdots & & \vdots\\\hline
F_{t}(\alpha_{1})e_{0}^{(1)} & \cdots & F_{t}(\alpha_{d_{0}})e_{0}^{(d_{0})}\\
\vdots & & \vdots\\
F_{t}(\alpha_{1})e_{\delta_{t}-1}^{(1)} & \cdots & F_{t}(\alpha_{d_{0}%
})e_{\delta_{t}-1}^{(d_{0})}\\\hline
\alpha_{1}^{0}\left(  x-\alpha_{1}\right)  & \cdots & \alpha_{{d_{0}}}%
^{0}\left(  x-\alpha_{{d_{0}}}\right) \\
\vdots & & \vdots\\
\alpha_{1}^{\varepsilon-2}\left(  x-\alpha_{1}\right)  & \cdots &
\alpha_{{d_{0}}}^{\varepsilon-2}\left(  x-\alpha_{{d_{0}}}\right)
\end{array}
\right]\]

\item Substitute \eqref{eqs:ej_exclude_i} into the above determinant and we get
\begin{align*}
\det \text{RHS}=    &(-1)^{\sum_{i=1}^{t}{(0+\cdots+\delta_{i}-1)}}\cdot a_{0d_{0}}^{\sum_{i=1}%
^{t}\delta_{i}}\cdot\\
&\qquad\det\left[
\begin{array}
[c]{ccc}%
F_{1}(\alpha_{1})(-1)^{0}e_{0}\alpha_{1}^{0} & \cdots & F_{1}(\alpha_{d_{0}})(-1)^{0}e_{0}\alpha_{d_{0}%
}^{0}\\
\vdots & & \vdots\\
F_{1}(\alpha_{1})\left(  \sum\limits_{k=0}^{\delta_{1}-1}(-1)^{k}e_{\delta
_{1}-1-k}\alpha_{1}^{k}\right)  & \cdots & F_{1}(\alpha_{d_{0}})\left(
\sum_{k=0}^{\delta_{1}-1}(-1)^{k}e_{\delta_{1}-1-k}\alpha_{d_{0}}^{k}\right)
\\\hline
\vdots &  & \vdots\\
\vdots & & \vdots\\\hline
(-1)^{0}F_{t}(\alpha_{1})\alpha_{1}^{0} & \cdots & (-1)^{0}F_{t}(\alpha_{d_{0}})\alpha_{d_{0}%
}\\
\vdots & & \vdots\\
F_{t}(\alpha_{1})\left(  \sum\limits_{k=0}^{\delta_{t}-1}(-1)^{k}e_{\delta
_{t}-1-k}\alpha_{1}^{k}\right)  & \cdots & F_{t}(\alpha_{d_{0}})\left(
\sum\limits_{k=0}^{\delta_{t}-1}(-1)^{k}e_{\delta_{t}-1-k}\alpha_{d_{0}}%
^{k}\right) \\\hline
\alpha_{1}^{0}\left(  x-\alpha_{1}\right)  & \cdots & \alpha_{{d_{0}}}%
^{0}\left(  x-\alpha_{{d_{0}}}\right) \\
\vdots & & \vdots\\
\alpha_{1}^{\varepsilon-2}\left(  x-\alpha_{1}\right)  & \cdots &
\alpha_{{d_{0}}}^{\varepsilon-2}\left(  x-\alpha_{{d_{0}}}\right)
\end{array}
\right]
\end{align*}

\item For each block of the form
\[
T_i:=\left[
\begin{array}
[c]{ccc}%
F_{i}(\alpha_{1})(-1)^{0}e_{0}\alpha_{1}^{0} & \cdots & F_{i}(\alpha_{d_{0}})(-1)^{0}e_{0}\alpha_{d_{0}%
}^{0}\\
\vdots &  & \vdots\\
F_{i}(\alpha_{1})\sum\limits_{k=0}^{\delta_{i}-1}(-1)^{k}e_{\delta_{i}%
-1-k}\alpha_{1}^{k} & \cdots & F_{i}(\alpha_{d_{0}})\sum\limits_{k=0}%
^{\delta_{i}-1}(-1)^{k}e_{\delta_{i}-1-k}\alpha_{d_{0}}^{k}%
\end{array}
\right]
\]
we subtract the $\ell$th row multiplied by $e_{j-\ell}$ where
$1\le\ell\le j-1$ successively from the $j$th row for $j=2,\ldots,\delta_{i}$.
\begin{align*}
T_i&=\left[
\begin{array}
[c]{ccc}%
F_{i}(\alpha_{1})(-1)^{0}e_{0}\alpha_{1}^{0} & \cdots & F_{i}(\alpha_{d_{0}})(-1)^{0}e_{0}\alpha_{d_{0}%
}^{0}\\
F_{i}(\alpha_{1})\sum\limits_{k=0}^{1}(-1)^{k}e_{1-k}\alpha_{1}^{k} & \cdots & F_{i}(\alpha_{d_{0}})\sum\limits_{k=0}%
^{1}(-1)^{k}e_{1-k}\alpha_{d_{0}}^{k}\\
\vdots &  & \vdots\\
F_{i}(\alpha_{1})\sum\limits_{k=0}^{\delta_{i}-1}(-1)^{k}e_{(\delta_{i}%
-1)-k}\alpha_{1}^{k} & \cdots & F_{i}(\alpha_{d_{0}})\sum\limits_{k=0}%
^{\delta_{i}-1}(-1)^{k}e_{(\delta_{i}-1)-k}\alpha_{d_{0}}^{k}%
\end{array}
\right]\\
&\qquad\text{substituting $e_0=1$ into the first row}\\
&=\left[
\begin{array}
[c]{ccc}%
(-1)^{0}F_{i}(\alpha_{1})\alpha_{1}^{0} & \cdots & (-1)^{0}F_{i}(\alpha_{d_{0}})\alpha_{d_{0}%
}^{0}\\
F_{i}(\alpha_{1})\left((-1)^{0}e_{1}\alpha_{1}^0+(-1)^{1}e_{0}\alpha_{1}^{1} \right)& \cdots & F_{i}(\alpha_{1})\left((-1)^{0}e_{1}\alpha_{d_0}^0+(-1)^{1}e_{0}\alpha_{d_0}^{1} \right)\\
\vdots &  & \vdots\\
F_{i}(\alpha_{1})\sum\limits_{k=0}^{\delta_{i}-1}(-1)^{k}e_{(\delta_{i}%
-1)-k}\alpha_{1}^{k} & \cdots & F_{i}(\alpha_{d_{0}})\sum\limits_{k=0}%
^{\delta_{i}-1}(-1)^{k}e_{(\delta_{i}-1)-k}\alpha_{d_{0}}^{k}%
\end{array}
\right]\\
&\qquad\text{subtracting the first row multiplied by $e_{1}$ from the second row and using $e_0=1$}\\
&\rightarrow \left[
\begin{array}
[c]{ccc}%
(-1)^{0}F_{i}(\alpha_{1})\alpha_{1}^{0} & \cdots & (-1)^{0}F_{i}(\alpha_{d_{0}})\alpha_{d_{0}%
}^{0}\\
(-1)^{1}F_{i}(\alpha_{1})\alpha_{1}^1 & \cdots & (-1)^{1}F_{i}(\alpha_{d_{0}})\alpha_{d_{0}}^1\\
\vdots &  & \vdots\\
F_{i}(\alpha_{1})\sum\limits_{k=0}^{\delta_{i}-1}(-1)^{k}e_{(\delta_{i}%
-1)-k}\alpha_{1}^{k} & \cdots & F_{i}(\alpha_{d_{0}})\sum\limits_{k=0}%
^{\delta_{i}-1}(-1)^{k}e_{(\delta_{i}-1)-k}\alpha_{d_{0}}^{k}%
\end{array}
\right]\\
&\qquad\text{now considering the third row and expanding it }\\
&=\left[
\begin{array}
[c]{ccc}%
(-1)^{0}F_{i}(\alpha_{1})\alpha_{1}^{0} & \cdots & (-1)^{0}F_{i}(\alpha_{d_{0}})\alpha_{d_{0}%
}^{0}\\
(-1)^{1}F_{i}(\alpha_{1})\alpha_{1}^1 & \cdots & (-1)^{1}F_{i}(\alpha_{d_{0}})\alpha_{d_{0}}^1\\
F_{i}(\alpha_{1})\left((-1)^{0}e_2\alpha_{1}^{0}+(-1)^{1}e_{1}\alpha_{1}^{1}+(-1)^{2}e_{0}\alpha_{1}^{2}\right) & \cdots & F_{i}(\alpha_{d_{0}})\left((-1)^{0}e_2\alpha_{d_0}^{0}+(-1)^{1}e_{1}\alpha_{d_0}^{1}+(-1)^{2}e_{0}\alpha_{d_0}^{2}\right)\\
\vdots &  & \vdots\\
F_{i}(\alpha_{1})\sum\limits_{k=0}^{\delta_{i}-1}(-1)^{k}e_{(\delta_{i}%
-1)-k}\alpha_{1}^{k} & \cdots & F_{i}(\alpha_{d_{0}})\sum\limits_{k=0}%
^{\delta_{i}-1}(-1)^{k}e_{(\delta_{i}-1)-k}\alpha_{d_{0}}^{k}%
\end{array}
\right]\\
&\qquad\text{subtracting the first two rows multiplied by $e_{2}$ and
$e_{1}$ respectively from the third row}\\
&\qquad\text{and using $e_0=1$}\\
&\rightarrow \left[
\begin{array}
[c]{ccc}%
(-1)^{0}F_{i}(\alpha_{1})\alpha_{1}^{0} & \cdots & (-1)^{0}F_{i}(\alpha_{d_{0}})\alpha_{d_{0}%
}^{0}\\
(-1)^{1}F_{i}(\alpha_{1})\alpha_{1}^1 & \cdots & (-1)^{1}F_{i}(\alpha_{d_{0}})\alpha_{d_{0}}^1\\
(-1)^{2}F_{i}(\alpha_{1})\alpha_{1}^2 & \cdots & (-1)^{2}F_{i}(\alpha_{d_{0}})\alpha_{d_{0}}^2\\
\vdots &  & \vdots\\
F_{i}(\alpha_{1})\sum\limits_{k=0}^{\delta_{i}-1}(-1)^{k}e_{(\delta_{i}%
-1)-k}\alpha_{1}^{k} & \cdots & F_{i}(\alpha_{d_{0}})\sum\limits_{k=0}%
^{\delta_{i}-1}(-1)^{k}e_{(\delta_{i}-1)-k}\alpha_{d_{0}}^{k}%
\end{array}
\right]\\
&\qquad \text{repeating the process and eventually obtaining the following matrix block:}\\
&\rightarrow\left[
\begin{array}
[c]{ccc}%
(-1)^{0}F_{i}(\alpha_{1})\alpha_{1}^{0} & \cdots & (-1)^{0}F_{i}(\alpha
_{d_{0}})\alpha_{d_{0}}^{0}\\
\vdots &  & \vdots\\
(-1)^{\delta_{i}-1}F_{i}(\alpha_{1})\alpha_{1}^{\delta_{i}-1} & \cdots &
(-1)^{\delta_{i}-1}F_{i}(\alpha_{d_{0}})\alpha_{d_{0}%
}^{\delta_{i}-1}%
\end{array}
\right]
\end{align*}
\end{enumerate}

\item Due to the multi-linear and anti-symmetry properties of determinant, $\det \text{RHS}$ in Step (6b) keeps unchanged after the transformations. That is,
\begin{align*}
\det \text{RHS}=(-1)^{\sum_{i=1}^{t}{(0+\cdots+\delta_{i}-1)}}\cdot a_{0d_{0}}^{\sum_{i=1}%
^{t}\delta_{i}}\cdot
  &  \det\left[
\begin{array}
[c]{ccc}%
(-1)^{0}F_{1}(\alpha_{1})\alpha_{1}^{0} & \cdots & (-1)^{0}F_{1}(\alpha_{d_{0}})\alpha_{d_{0}%
}^{0}\\
\vdots & & \vdots\\
(-1)^{\delta_1-1}F_{1}(\alpha_{1})\alpha_{1}^{\delta_{1}-1} & \cdots & (-1)^{\delta_1-1}F_{1}(\alpha_{d_{0}%
})\alpha_{d_{0}}^{\delta_{1}-1}\\\hline
\vdots & & \vdots\\
\vdots &  & \vdots\\\hline
(-1)^{0}F_{t}(\alpha_{1})\alpha_{1}^{0} & \cdots & (-1)^{0}F_{t}(\alpha_{d_{0}})\alpha_{d_{0}%
}^{0}\\
\vdots &  & \vdots\\
(-1)^{\delta_t-1}F_{t}(\alpha_{1})\alpha_{1}^{\delta_{t}-1} & \cdots & (-1)^{\delta_t-1}F_{t}(\alpha_{d_{0}%
})\alpha_{d_{0}}^{\delta_{t}-1}\\\hline
\alpha_{1}^{0}\left(  x-\alpha_{1}\right)  & \cdots & \alpha_{{d_{0}}}%
^{0}\left(  x-\alpha_{{d_{0}}}\right) \\
\vdots &  & \vdots\\
\alpha_{1}^{\varepsilon-2}\left(  x-\alpha_{1}\right)  & \cdots &
\alpha_{{d_{0}}}^{\varepsilon-2}\left(  x-\alpha_{{d_{0}}}\right)
\end{array}
\right]
\end{align*}

\item Next extract the factor $(-1)^{k_i}$ from the $i$th row ($1\le i\le \delta_1+\cdots+\delta_t$) where $k_i$ is some natural number depending on $i$:
\begin{align*}
\det \text{RHS} &=(-1)^{2\sum_{i=1}^{t}{(0+\cdots+\delta_{i}-1)}}\cdot a_{0d_{0}}^{|\delta|}\cdot  \det\left[
\begin{array}
[c]{ccc}%
F_{1}(\alpha_{1})\alpha_{1}^{0} & \cdots & F_{1}(\alpha_{d_{0}})\alpha_{d_{0}%
}^{0}\\
\vdots & & \vdots\\
F_{1}(\alpha_{1})\alpha_{1}^{\delta_{1}-1} & \cdots & F_{1}(\alpha_{d_{0}%
})\alpha_{d_{0}}^{\delta_{1}-1}\\\hline
\vdots &  & \vdots\\
\vdots &  & \vdots\\\hline
F_{t}(\alpha_{1})\alpha_{1}^{0} & \cdots & F_{t}(\alpha_{d_{0}})\alpha_{d_{0}%
}^{0}\\
\vdots &  & \vdots\\
F_{t}(\alpha_{1})\alpha_{1}^{\delta_{t}-1} & \cdots & F_{t}(\alpha_{d_{0}%
})\alpha_{d_{0}}^{\delta_{t}-1}\\\hline
\alpha_{1}^{0}\left(  x-\alpha_{1}\right)  & \cdots & \alpha_{{d_{0}}}%
^{0}\left(  x-\alpha_{{d_{0}}}\right) \\
\vdots &  & \vdots\\
\alpha_{1}^{\varepsilon-2}\left(  x-\alpha_{1}\right)  & \cdots &
\alpha_{{d_{0}}}^{\varepsilon-2}\left(  x-\alpha_{{d_{0}}}\right)
\end{array}
\right]
\end{align*}

\item By Lemma \ref{lem:equiv},
$$\det \text{RHS}=(-1)^{2\sum_{i=1}^{t}{(0+\cdots+\delta_{i}-1)}}\cdot a_{0d_{0}}^{|\delta|}\cdot a_{0d_{0}}^{-\delta_{0}}\cdot S_{\delta}\cdot V=a_{0d_{0}}^{-\delta_0+|\delta|}\cdot S_{\delta}\cdot V$$ Recall that in Step 5, we have shown that $$\det \text{LHS}=\det M_{\delta}^{\text{B\'{e}zout}}\cdot V$$
\item Finally, comparing the two expressions, we have
\[
S_{\delta}=a_{0d_{0}}^{\delta_0-|\delta|}\cdot\det M_{\delta
}^{\text{B\'{e}zout}}
\]

\end{enumerate}
\end{proof}

\section{Conclusion}

\label{sec:conclusion} In this paper, we proposed a definition of a
subresultant for several univariate polynomials in terms of roots. To show
that the definition is meaningful and useful, we presented two fundamental applications:
parametric gcd and parametric multiplicity. For  computation, we gave three different expressions of the  subresultants in terms of coefficients.
A natural challenge for future work is to  generalize further to arbitrary number
of multivariate and/or Ore polynomials.

\medskip\noindent\textbf{Acknowledgements.} Hoon Hong's work was 
supported by National Science Foundations of USA (Grant Nos: 2212461 and 1813340).
 Jing Yang's work was
supported by National Natural Science Foundation of China (Grant No. 12261010).


\end{document}